\documentclass[sigconf,nonacm]{acmart}

\setcopyright{none}
\AtBeginDocument{%
  }

\newcommand{\Fp}{\mathbb{F}_p}

\newcommand{\Obl}{\mathrm{Obl}}

\usepackage{enumitem}

\usepackage{tikz}
\usetikzlibrary{positioning,arrows.meta,calc,backgrounds,decorations.pathreplacing}

\newcounter{algorithm}
\newcounter{algoline}[algorithm]
\newcommand{\algcaption}[1]{%
  \refstepcounter{algorithm}%
  \noindent\textbf{Algorithm~\thealgorithm.}~#1\par\smallskip%
}
\newcommand{\Require}{\par\noindent\textbf{Input:}\ \ignorespaces}
\newcommand{\Ensure}{\par\noindent\textbf{Output:}\ \ignorespaces}
\newcommand{\State}{%
  \par\refstepcounter{algoline}%
  \noindent\makebox[2.4em][l]{\footnotesize\arabic{algoline}:}\ignorespaces}
\newcommand{\Return}{\textbf{return}\ }

\newcommand{\reved}[1]{#1}

\newcommand{\F}{Fig.}
\newcommand{\Def}{Def.}
\newcommand{\Alg}{Alg.}

\newcommand{\parh}[1]{\smallskip\noindent\textbf{#1}}

\newsavebox{\rqanswerbox}
\newenvironment{rqanswer}{%
  \par\smallskip
  \begin{lrbox}{\rqanswerbox}%
  \begin{minipage}{\dimexpr\columnwidth-2\fboxsep-2\fboxrule\relax}%
}{%
  \end{minipage}%
  \end{lrbox}%
  \noindent\fbox{\usebox{\rqanswerbox}}%
  \par\smallskip
}

\theoremstyle{definition}
\newtheorem{definition}{Definition}

\begin{document}

\title{Sound Debloating of Redundant Checks in Zero-Knowledge Machine-Learning Circuits}
\titlenote{This is the authors' version for ACM CCS~2026, with additional appendices.
Conference version: Zhantong Xue, Pingchuan Ma, Zhaoyu Wang, Yuguang Zhou,
Huaijin Wang, and Shuai Wang. 2026. Sound Debloating of Redundant Checks in
Zero-Knowledge Machine-Learning Circuits. In \emph{Proceedings of the 2026
ACM SIGSAC Conference on Computer and Communications Security (CCS~'26)},
November 15--19, 2026, The Hague, Netherlands. ACM. Forthcoming.
Assigned DOI: \url{https://doi.org/10.1145/3830454.3846774}.
\copyright~2026 the authors. This manuscript is licensed under
\href{https://creativecommons.org/licenses/by/4.0/}{CC BY 4.0}.}

\author{Zhantong Xue}
\orcid{0009-0000-3419-9431}
\affiliation{%
  \institution{Hong Kong University of Science and Technology}
  \city{Hong Kong}
  \country{China}
}
\email{zxueai@cse.ust.hk}

\author{Pingchuan Ma}
\authornote{Corresponding author.}
\orcid{0000-0001-7680-2817}
\affiliation{%
  \institution{Zhejiang University of Technology}
  \city{Hangzhou}
  \country{China}
}
\email{pma@zjut.edu.cn}

\author{Zhaoyu Wang}
\orcid{0009-0009-6892-1264}
\affiliation{%
  \institution{Hong Kong University of Science and Technology}
  \city{Hong Kong}
  \country{China}
}
\email{zwangjz@cse.ust.hk}

\author{Yuguang Zhou}
\orcid{0000-0001-5069-4563}
\affiliation{%
  \institution{Hong Kong University of Science and Technology}
  \city{Hong Kong}
  \country{China}
}
\email{yzhougv@cse.ust.hk}

\author{Huaijin Wang}
\orcid{0000-0002-1066-0331}
\affiliation{%
  \institution{Shandong University}
  \city{Jinan}
  \country{China}
}
\email{huaijinwang@sdu.edu.cn}

\author{Shuai Wang}
\orcid{0000-0002-0866-0308}
\affiliation{%
  \institution{Hong Kong University of Science and Technology}
  \city{Hong Kong}
  \country{China}
}
\email{shuaiw@cse.ust.hk}

\renewcommand{\shortauthors}{Xue et al.}

\begin{abstract}
Zero-knowledge (ZK) proof systems for neural-network inference compile the model
into a system of arithmetic constraints. Many of these constraints are redundant
checks: range proofs, sign lookups, and bit decompositions who are globally
entailed by the rest of the circuit through chains of reasoning that span
distant gadgets. Removing them shrinks the circuit and accelerates proving, but
the removal must be carefully justified: an unsoundly debloated circuit becomes
forgeable, accepting witnesses the original would have rejected and so allowing
a prover to claim, for example, that a neural network produced an output it
never actually computed. Such soundness vulnerabilities are not hypothetical:
under-constrained circuits in deployed ZK systems have enabled attackers to
forge transactions and bypass verification entirely.

We present an automated framework that removes redundant checks while provably
preserving soundness. For each candidate removal, our tool first checks whether
the rest of the circuit, on its own, can still rule out every value the removed
check was excluding. Using whole-circuit abstract interpretation, the analysis
searches for such alternative justifications and records them in a provenance
graph; a check is then removed only when an alternative path through the graph
still derives the facts that it is checking. This ensures that the debloated
circuit opens no new forging strategy to an adversary. We evaluate circuits
spanning MLP, CNN, RNN, and transformer architectures generated by two
production frameworks (ezkl and zkml), with up to 25.3 million constraints. Our
tool removes up to 48.7\% of constraints and reduces prover time by up to
72.8\%, without weakening security.
\end{abstract}



\maketitle

\section{Introduction}
\label{sec:intro}


Zero-knowledge (ZK) proof systems~\cite{goldwasser85} let a prover convince a
verifier that a computation was performed correctly without revealing its
inputs. The computation is encoded as \reved{a \emph{constraint system}
$\mathcal{C}$, a quantifier-free conjunction of polynomial equations over a
finite field; the prover's secret \emph{witness} $w$ is a satisfying
assignment, $w \models \mathcal{C}$}. A fast-growing application is
\emph{zero-knowledge machine learning}: \reved{a framework compiles inference of a
trained network into such a constraint system, so that a witness
satisfies $\mathcal{C}$ exactly when it encodes a correct inference}. A
proof of knowledge of such a witness therefore certifies the inference
while revealing neither the input nor the model
weights~\cite{peng2026surveyzeroknowledgeproofbased}, enabling on-chain AI
verification and privacy-preserving model serving.\looseness=-1

Prover cost scales with the number of constraints. Even small models compile to
circuits with tens to hundreds of thousands of constraints, and proving a single
inference can take minutes to hours~\cite{kang2022scalingtrustlessdnninference,
ezkl}. Recent benchmarks report proving overhead of $10{,}000\times$ or more
over plain inference, with hundreds of gigabytes of RAM consumed for
moderate-sized models~\cite{moduluslabs2023costofintelligence}. Shrinking the
constraint system is thus a natural target for reducing this
cost.\looseness=-1

A major cause of this overhead, accounting for up to 48.7\% of constraints in
our evaluation (\S\ref{sec:eval}), is the large number and diversity of
well-formedness \emph{checks} that ZK machine learning (ML) frameworks emit on
intermediate variables, including range checks on activation outputs, sign
lookups for value decompositions, remainder bounds in integer division gadgets,
and bit-decomposition validity proofs~\cite{ezkl, zkml}. \reved{Like assertions in
program verification, checks compute no value; they only restrict the
admissible assignments}, guarding against a malicious prover supplying an
out-of-range witness. A \emph{gadget} is a reusable subcircuit implementing
a single operation, such as an activation or division. Yet composition makes
many of these checks redundant: each gadget is compiled as a self-contained
template carrying every check needed for any input it might receive in
isolation. When the framework wires gadgets together, an upstream check often
constrains a shared variable to a range that already implies a downstream
check on the same variable.\looseness=-1

For example, the constraint emitted by a ReLU gadget bounds its output $z$ to
$[0, R_{\mathsf{relu}}]$, while the constraint from the next layer's range check
independently bounds the same $z$ to $[0, R_{\mathsf{rng}}]$ with
$R_{\mathsf{rng}} \le R_{\mathsf{relu}}$; the upstream constraint still asserts
a valid bound, but that bound is already implied by the tighter downstream one.
Such compositional redundancies pervade ML circuits because the same gadget
templates are instantiated hundreds of times across layers
(\S\ref{sec:motivation}). This raises a natural question:\looseness=-1
\begin{quote}
\emph{Can we automatically identify and remove (i.e., debloat\footnote{Analogous
to software debloating, which removes unused code from binaries.}) these
redundant checks to reduce proving cost, without compromising security?}
\end{quote}

\parh{Debloating \& ZK ML security.}
Answering this requires defining what makes debloating safe. We borrow
``\emph{sound}'' from ZK proof systems, where it forbids any prover from
convincing the verifier of a false statement: \reved{a debloated constraint set
$\mathcal{C}' \subseteq \mathcal{C}$ is \emph{sound} when every witness it accepts is also accepted
by the original ($w \models \mathcal{C}'$ implies $w \models \mathcal{C}$);
equivalently, every removed check $c$ is \emph{entailed} by the survivors,
i.e., a logical consequence of them ($\mathcal{C}' \vdash c$), leaving the
model set unchanged.}
\emph{Unsound} debloating enlarges the accepted-witness set, producing exactly
the \emph{under-constrained} circuits behind real-world ZK
exploits~\cite{aztec2022disclosure, cve-2019-7167}. As \F~\ref{fig:threat-model}
illustrates, this lets a malicious prover supply a forged witness for the model
$\mathcal{M}$ on input $x$, producing a proof~$\pi$ that the verifier still
accepts and claiming inferences $\mathcal{M}$ never produced. Such attacks
evade detection, since $\pi$ is the only artifact the verifier ever sees.
Soundness is therefore critical for any optimisation that modifies the
constraint set.\looseness=-1

\input{figures/threat-model}

\parh{Challenges.} We observe that soundly eliminating redundant checks is
challenging for ZK ML frameworks. First, redundancy is inherently \emph{global}:
a check that looks necessary in a local gadget may be implied by constraints
several layers away, with the implication flowing in either direction across the
circuit, so no per-gadget or single-direction analysis suffices. Second,
redundancy does not compose. Two checks may each be individually redundant yet
unsound to remove together, because each derivation of redundancy uses the
other removed check as a premise; we call this circularity the
\emph{premise-asymmetry} problem (\S\ref{sec:obstacles}). Third, exact
reasoning is intractable: \reved{each redundancy query $\mathcal{C} \setminus \{c\}
\vdash c$ is an entailment in nonlinear field arithmetic (the analogue of
QF\_NIA), already co-NP-hard for a single
constraint (\S\ref{sec:obstacles})}, and the largest sound removal is harder
still.\looseness=-1

\parh{Our solution.} \reved{This \emph{semantic} redundancy, checks that are active
and unique yet entailed by the rest of the circuit, is invisible to existing
optimizers, which decide removals from the syntactic or algebraic form of
constraints alone (\S\ref{sec:prior-opt}); on our benchmarks they remove at
most 0.4\% (\S\ref{sec:eval-rq3}).} We give, to our knowledge, the first
sound and scalable method for this class. The technical core is a
provenance-tracked proof search. We pair \emph{abstract interpretation}~\cite{cousot77}, a
tractable analysis that derives sound \emph{facts} (abstract properties such as
range bounds and known bit patterns) about each variable, with a
\emph{provenance graph}, a proof log recording, for every derived fact, the
constraints and premises that derived it. A check is removed
only when every fact it enforces remains derivable from the remaining
constraints, decided by a proof search over the graph; cycle detection in the
search rules out the circular mutual derivations that premise asymmetry would
otherwise allow. The successful search itself certifies the soundness of each
removal covered by the entailment rules. In our evaluation
(\S\ref{sec:eval}), the tool eliminates 25.5\% of constraints on average (up
to 48.7\%) and reduces proving time by 33.4\% on average (up to 72.8\%).\looseness=-1

\parh{Contributions.} In sum, we make the following contributions:

\begin{itemize}[leftmargin=*]
  \item We identify the well-formedness checks emitted by ZK machine learning
    frameworks as \reved{a systematic source of semantic redundancy untapped by
    existing optimizers}, and formalize the three challenges that make its
    sound removal non-trivial.

  \item We propose an abstract interpretation framework that derives statically
    checkable facts about each variable and tracks the provenance of each fact
    back to the constraints that proved it. The construction is sound: for each
    removed constraint, no malicious prover gains additional power from the
    removal.

  \item We implement the approach in a tool and evaluate it on circuits spanning
    MLP, CNN, RNN, and transformer architectures from two production frameworks.
    The tool eliminates 25.5\% of constraints on average (up to 48.7\%), reduces
    proving time by 33.4\% on average (up to 72.8\%), and scales to circuits
    with up to 25.3 million constraints.
\end{itemize}

\section{Background}
\label{sec:background}

\subsection{Constraints in ZK Proof Systems}

A zero-knowledge proof convinces a verifier that a prover knows a \emph{witness}
satisfying some computation, without revealing the witness~\cite{goldwasser85}.
Modern proof systems make this checkable by compiling the computation into a
\emph{constraint system}~$\mathcal{C}$: a finite conjunction of algebraic
equations over a prime field~$\Fp$. \reved{In logical terms, $\mathcal{C}$ is a
quantifier-free formula, a witness $w$ is a model of it ($w \models
\mathcal{C}$), and the proof certifies
satisfiability without disclosing the model. (The \emph{protocol}'s
cryptographic soundness is orthogonal to this paper: our subject is which
models $\mathcal{C}$ itself admits.)} Both prover time and memory scale with
the number of constraints, so
reducing constraint count is a primary lever for accelerating ZK systems.\looseness=-1

Two representations dominate practice, and our tool supports both. \reved{Both
constrain assignments $w : \mathcal{V} \to \Fp$ over a variable set
$\mathcal{V}$, differing in the shape of an individual constraint.}

\emph{R1CS}~\cite{snarksforc} \reved{states each constraint as a quadratic
identity: a product of two linear forms over $\mathcal{V}$ equals a
third, so every constraint has degree at most two}. It is the native format for
Groth16~\cite{groth16} and front-ends such as
Circom~\cite{belles2023circomacircuitdescription} and
ZoKrates~\cite{zokrates}.\looseness=-1

\emph{PLONKish}, introduced by PLONK~\cite{plonk} and adopted by
halo2~\cite{halo2}, instead lays the variables out in a table whose columns
are partitioned into private \emph{advice} columns (holding witness
variables) and public \emph{fixed} columns (compile-time constants, including
the \emph{selectors} that switch individual constraints on and off per row).
\reved{It has three constraint forms: a \emph{gate} requires a polynomial over
the variables of a row to vanish on every row where its selector is
active; a \emph{lookup} requires a tuple of variables to appear in a
predefined table; and a \emph{copy constraint} equates two cells.} Lookups and custom gates make table-driven nonlinearities cheap, which
is why general-purpose ZK machine learning frameworks~\cite{ezkl, zkml}
target PLONKish.\looseness=-1

\subsection{ZK Machine Learning Frameworks}

Standard ML frameworks such as PyTorch~\cite{pytorch} and
TensorFlow~\cite{tensorflow} execute a model's forward pass over IEEE
floating-point tensors and emit whatever output the computation produces. ZK ML
reframes that same forward pass as a verifiable relation: the prover convinces a
verifier that an input-weight-output triple satisfies $y = f(x; W)$ over $\Fp$,
typically with the weights kept
private~\cite{peng2026surveyzeroknowledgeproofbased}. We study two
general-purpose ZK ML frameworks of this kind: \texttt{ezkl}~\cite{ezkl} and
\texttt{zkml}~\cite{zkml}. Each ingests a trained model (\texttt{ezkl}: ONNX;
\texttt{zkml}: TFLite) and \reved{compiles every layer or operator as a gadget
(\S\ref{sec:intro}): each use emits a fresh constraint set, and the circuit
$\mathcal{C}$ is the union of these per-gadget sets, adjacent gadgets
sharing the variables that carry values between them}.\looseness=-1

Two restrictions distinguish ZK ML from standard ML toolchains: (1)
the constraints live in a prime field $\Fp$ with no IEEE floats, so \reved{the
toolchain \emph{quantizes}: a real value $x$ is represented by the integer
$\lfloor x \cdot 2^s \rceil$ for a scale exponent $s$, integers
enter $\Fp$ through their canonical signed representatives, and every
fixed-point operation is re-expressed over $\Fp$}; and (2) soundness
requires every inference step to be constrained: any operation left
unconstrained lets a malicious prover assign its witness arbitrarily. These
restrictions also change the target of analysis: rather than the real-valued
network of DNN verification (\S\ref{sec:related}), our object is the
field-arithmetic system $\mathcal{C}$ the compiler emits, and the question
shifts from input-output properties to entailment among its
constraints.\looseness=-1

ZK ML frameworks therefore scaffold the model with three classes of
\emph{checking constraints}: \reved{compiler-emitted invariants pinning each witness
variable to the bounded range where field arithmetic agrees with the intended
fixed-point semantics (otherwise wrapped-around field values satisfy the same
equations, admitting unintended models)}. \emph{Range checks}
prove that a witness value lies in a specified
interval, bounding activations to their declared bit-width and constraining the
auxiliary digits introduced by other gadgets. \emph{Value decomposition}
realizes integer division by a constant $B$ for fixed-point rescaling: the
compiler introduces auxiliary witness variables $(q, r)$ with $v = qB + r$ and
$r$ range-checked into $[0, B)$, since $\Fp$ has no native integer division.
\emph{Lookup-table non-linearities} realize ReLU, Sigmoid, Softmax, and similar
functions by precomputing input-output pairs across the quantized domain and
asserting that each pair appears in the table.\looseness=-1

\texttt{ezkl} and \texttt{zkml} differ in how these checks are emitted.
\texttt{ezkl} factors them through shared infrastructure: every range proof
goes through a single range-check service, rescaling decomposes each value into
base-$B$ digits (default $B = 16384$) whose ranges are checked via that same
service, and activation lookups consult their tables directly, with input bounds
inherited from the upstream decomposition. \texttt{zkml} inlines them into each gadget:
range-check lookups are emitted at every call site, rescaling uses a single
division gadget without digit decomposition, and every activation lookup is
preceded by a separate input-range lookup. \texttt{ezkl}'s scaffolding is thus
\emph{factored} across the circuit; \texttt{zkml}'s is \emph{replicated per
gadget}.\looseness=-1

\subsection{Constraint Optimization in ZK Systems}
\label{sec:prior-opt}

Several prior tools reduce or analyze the constraint count of ZK circuits.
Table~\ref{tab:prior-opt} compares them along the dimensions that determine
whether they detect the redundancies our work targets.

\begin{table}[t]
\centering
\caption{Capability comparison of prior tools.}
\label{tab:prior-opt}
\small
\setlength{\tabcolsep}{4pt}
\resizebox{\columnwidth}{!}{%
\begin{tabular}{lcccccc}
\toprule
                    & Compilers  & Distilling    & CirC       & Clap       & Halo2-Analyzer \\
\midrule
Algebraic optimization & \checkmark & \checkmark & \checkmark & \checkmark & $\times$     \\
Syntactic optimization & \checkmark & \checkmark & \checkmark & \checkmark & $\checkmark$    \\
Semantic optimization & $\times$    & $\times$    & $\times$    & $\times$   & $\times$ \\
Logical entailment  & $\times$    & $\times$    & $\times$    & $\times$    & $\times$   \\
\bottomrule
\end{tabular}%
}%
\end{table}

\parh{Compiler optimizations.}
Mainstream ZK compilers (Circom~\cite{belles2023circomacircuitdescription}
(\texttt{O2}), ZoKrates~\cite{zokrates}, Noir~\cite{noir-lang}, and
Leo~\cite{leo-lang}) apply standard syntactic and algebraic simplifications,
including constant folding, common subexpression elimination, and Gaussian
elimination over linear constraints. Nonlinear constraints are untouched, and
none of these passes reasons about the values variables may take.\looseness=-1

\parh{Distilling constraints.}
Albert et al.~\cite{distillconstraints} extend Circom by deriving \emph{new}
linear consequences from nonlinear R1CS triples via Gaussian elimination over
their quadratic-monomial coefficients, then substituting and deleting. This
analysis is purely algebraic.\looseness=-1

\parh{CirC.}
CirC~\cite{Ozdemir2022CirCCI} shares an IR across R1CS, SMT, and ILP back-ends
and applies standard IR-level optimizations during lowering: constant folding,
common-subexpression elimination, and linearity reduction. It has no awareness
of entailment redundancies.\looseness=-1

\parh{Clap.}
Clap~\cite{stronati2024clapsemanticpreservingoptimizingedsl} targets PLONKish
circuits with structural rewrites (linear inlining, gate consolidation, and
custom-gate synthesis) under formal soundness guarantees. The rewrites operate
on circuit syntax, so a check that is semantically entailed by other gadgets
remains invisible to them.\looseness=-1

\parh{Halo2 analyzer.}
The Halo2 analyzer\cite{halo2-analyzer} applies abstract interpretation to
PLONKish circuits to flag unused gates, unused columns, and
assigned-but-unconstrained cells. This analysis only detects structural dead
weight.

\reved{In sum, each prior pass decides removals from the syntactic or algebraic form
of constraints considered in isolation; none asks what the remaining
constraints entail. \S\ref{sec:eval-rq3} confirms the gap empirically: on
our benchmarks these tools remove at most 0.4\%, while entailment-based
debloating removes up to 48.7\%.}\looseness=-1

\section{Problem and Motivation}
\label{sec:motivation}

This section fixes the sound debloating problem formally
(\S\ref{sec:problem-formulation}), grounds it in a concrete zero-knowledge ML
fragment (\S\ref{sec:motivating-example}), and uses that fragment to surface
three challenges any sound debloater must address (\S\ref{sec:obstacles}).\looseness=-1

\subsection{Problem Formulation}
\label{sec:problem-formulation}


Let $\mathcal{C} = \{c_1, \ldots, c_n\}$ be a constraint system over variables
$\mathcal{V}$.

\begin{definition}[Satisfaction]
\label{def:satisfaction}
A witness $w : \mathcal{V} \to \Fp$ \emph{satisfies} a constraint $c$, \reved{written
$w \models c$, if $c$ holds under the assignment $w$}. We write $w \models
\mathcal{C}$ if $w \models c$ for every $c \in \mathcal{C}$.
\end{definition}

\begin{definition}[Entailment]
\label{def:entailment}
A set of constraints $\Gamma \subseteq \mathcal{C}$ (the \emph{premises})
\emph{entails} a constraint $c$, \reved{written $\Gamma \vdash c$, if every witness
satisfying $\Gamma$ also satisfies $c$}. We extend entailment to sets on the
right: for a set $\Phi$ of constraints, $\Gamma \vdash \Phi$ iff $\Gamma
\vdash c$ for every $c \in \Phi$.
\end{definition}

We define redundancy and sound debloating using entailment.

\begin{definition}[Redundancy]
\label{def:redundancy}
A constraint $c \in \mathcal{C}$ is \emph{redundant} with respect to
$\mathcal{C}$ if $\mathcal{C} \setminus \{c\} \vdash c$.
\end{definition}


\begin{definition}[Safe debloating]
\label{def:debloating}
A subset $\mathcal{R} \subseteq \mathcal{C}$ is a \emph{safe debloating} of
$\mathcal{C}$ if $\mathcal{C} \setminus \mathcal{R} \vdash \mathcal{R}$. This
is \emph{sound} since every witness that satisfies the debloated system
$\mathcal{C} \setminus \mathcal{R}$ also satisfies the original $\mathcal{C}$.
This property is stringent: an unsound removal admits a witness $w$ with $w
\models \mathcal{C} \setminus \mathcal{R}$ but $w \not\models \mathcal{C}$,
potentially allowing a malicious prover to forge a witness that the verifier
accepts but that does not correspond to a legitimate computation.
\end{definition}


The ideal target is a \emph{maximum} safe debloating:
\[
  \mathcal{R}^\star \;\in\; \arg\max_{\mathcal{R} \subseteq \mathcal{C}}\; |\mathcal{R}|
  \quad\text{subject to}\quad \mathcal{C} \setminus \mathcal{R} \vdash \mathcal{R}.
\]
Computing $\mathcal{R}^\star$ exactly is intractable (\S\ref{sec:obstacles}), so
we instead compute a large feasible $\mathcal{R}$ in polynomial time
(\S\ref{sec:approach}).\looseness=-1


\subsection{Motivating Example}
\label{sec:motivating-example}

\begin{figure}[t]

\resizebox{\columnwidth}{!}{%
\begin{tikzpicture}[
  font=\small,
  >={Latex[length=1.8mm, width=1.3mm]},
  box/.style={fill=black!6, text width=6.2cm,
              inner xsep=4pt, inner ysep=5pt,
              align=left},
  sub/.style={->, dashed, semithick},
]


\node[box, anchor=north west] (fc) at (0, 0) {%
  {\centering\textbf{FC}\par}\vspace{2pt}%
  \parbox[t]{3.1cm}{\centering\footnotesize $y = Wx + b$}%
  \parbox[t]{3.1cm}{\centering\footnotesize $c_{\mathsf{fc}}\colon y = Wx + b$}%
};

\node[box, below=3mm of fc] (relu) {%
  {\centering\textbf{ReLU}\par}\vspace{2pt}%
  \parbox[t]{3.1cm}{\centering\footnotesize $z = \max(0,\, y)$}%
  \parbox[t]{3.1cm}{\centering\footnotesize $c_{\mathsf{relu}}\colon z \in [0,\; R_{\mathsf{relu}}]$}%
};

\node[box, below=3mm of relu] (dec) {%
  {\centering\textbf{Decomposition}\par}\vspace{2pt}%
  \parbox[t]{3.1cm}{\centering\footnotesize $z = s(d_1 B + d_0)$}%
  \parbox[t]{3.1cm}{\centering\footnotesize
    $c_{\mathsf{dec}}\colon z = s(d_1 B{+}d_0)$\\[1pt]
    $c_{d_0}\colon d_0 \in [0,\, B)$\\[1pt]
    $c_{d_1}\colon d_1 \in [0,\, B)$\\[1pt]
    $c_{\mathsf{sign}}\colon s \in \{-1,{+}1\}$}%
};

\node[box, below=3mm of dec] (nxt) {%
  {\centering\textbf{Next layer}\par}\vspace{2pt}%
  \parbox[t]{3.1cm}{\centering\footnotesize\strut}%
  \parbox[t]{3.1cm}{\centering\footnotesize $c_{\mathsf{rng}}\colon z \in [0,\; R_{\mathsf{rng}}]$}%
};

\foreach \b in {fc, relu, dec, nxt} {
  \draw[dashed, very thin, black!40]
    ([yshift=-5mm] \b.north) -- ([yshift=1.5mm] \b.south);
}

\foreach \upper/\lower in {fc/relu, relu/dec, dec/nxt} {
  \path (\upper.south) -- (\lower.north) coordinate[pos=0.2] (arr@top)
                                          coordinate[pos=0.8] (arr@bot);
  \draw[line width=0.4pt, black, fill=white]
    ([xshift=-0.35mm] arr@top) --
    ([xshift=-0.35mm, yshift=0.6mm] arr@bot) --
    ([xshift=-0.8mm, yshift=0.6mm] arr@bot) --
    (arr@bot) --
    ([xshift=0.8mm, yshift=0.6mm] arr@bot) --
    ([xshift=0.35mm, yshift=0.6mm] arr@bot) --
    ([xshift=0.35mm] arr@top) -- cycle;
}

\coordinate (sub-start) at ([xshift=1.5mm] nxt.east);
\coordinate (sub-end)   at ([xshift=1.5mm] relu.east);
\coordinate (sub-mid)   at ([xshift=8mm] $(relu.east)!0.55!(nxt.east)$);
\draw[sub]
  (sub-start) -- (sub-start -| sub-mid) -- (sub-end -| sub-mid) -- (sub-end);
\node[anchor=north, rotate=90, font=\footnotesize\itshape, text=black!70]
  at ([xshift=0.5mm] sub-mid)
  {$c_{\mathsf{rng}} \Rightarrow c_{\mathsf{relu}}$ \;\scriptsize if $R_{\mathsf{relu}} \ge R_{\mathsf{rng}}$};

\coordinate (bic-top) at ([xshift=-6mm, yshift=-7mm] dec.north east);
\coordinate (bic-bot) at ([xshift=-6mm, yshift=-17mm] dec.north east);
\draw[decorate, decoration={brace, amplitude=3pt}, thick, black!60]
  (bic-top) -- (bic-bot)
  node[midway, right=4pt+1.5em, font=\footnotesize\itshape, text=black!70,
       rotate=90, anchor=south] {bound-induced collapse};

\end{tikzpicture}%
}

\par\smallskip
{\small \textit{(a) Multi-layer fragment with constraints per gadget.}}

\vspace{3mm}

\resizebox{\columnwidth}{!}{%
\begin{tikzpicture}[
  x=0.91cm,
  font=\footnotesize,
  >={Latex[length=1.5mm, width=1.1mm]},
  cn/.style={anchor=west, inner sep=1.5pt, font=\footnotesize},
  fact/.style={draw, rounded corners, fill=black!4,
               inner sep=3pt, anchor=west, font=\footnotesize},
  direct/.style={->, thick},
  joint/.style={->, semithick, dashed},
  title/.style={font=\footnotesize\itshape, anchor=south},
]

\node[title] at (1.8, 0.9) {Derivations of $x \in [0,\, 10]$};

\node[cn] at (0, 0.4) (L-c2) {$c_2\colon x \in [0,10]$};
\node[fact] at (2.6, 0.4) (L-fx) {$x \in [0,\, 10]$};
\draw[direct] (L-c2.east) -- node[above, font=\scriptsize\itshape] {direct} (L-fx.west);

\node[cn] at (0, -0.05) (L-c1) {$c_1\colon x{+}y{=}10$};
\node[cn] at (0, -0.45) (L-c3) {$c_3\colon y \in [0,10]$};
\node[fact] at (2.6, -0.25) (L-fxj) {$x \in [0,\, 10]$};
\draw[joint] (L-c1.east) -- node[above, font=\scriptsize\itshape, pos=0.4] {joint} (L-fxj.west);
\draw[joint] (L-c3.east) -- (L-fxj.west);


\node[title] at (6.0, 0.9) {Derivations of $y \in [0,\, 10]$};

\node[cn] at (4.2, 0.4) (R-c3) {$c_3\colon y \in [0,10]$};
\node[fact] at (6.8, 0.4) (R-fy) {$y \in [0,\, 10]$};
\draw[direct] (R-c3.east) -- node[above, font=\scriptsize\itshape] {direct} (R-fy.west);

\node[cn] at (4.2, -0.05) (R-c1) {$c_1\colon x{+}y{=}10$};
\node[cn] at (4.2, -0.45) (R-c2) {$c_2\colon x \in [0,10]$};
\node[fact] at (6.8, -0.25) (R-fyj) {$y \in [0,\, 10]$};
\draw[joint] (R-c1.east) -- node[above, font=\scriptsize\itshape, pos=0.4] {joint} (R-fyj.west);
\draw[joint] (R-c2.east) -- (R-fyj.west);

\end{tikzpicture}%
}

\par\smallskip
{\small \textit{(b) Premise asymmetry: removing both lookups breaks all joint paths.}}

\caption{Illustrative redundancy patterns in ZK ML circuits.}
\Description{Two panels illustrating redundant constraints. Panel (a)
  follows a value through fully connected, ReLU, decomposition, and next-layer
  gadgets. A next-layer range constraint subsumes the ReLU range constraint,
  while tighter bounds collapse decomposition digits and the sign. Panel (b)
  shows that each of the facts x in zero to ten and y in zero to ten has both a
  direct derivation and a joint derivation through the sum constraint and the
  other variable's range; removing both direct lookups breaks the joint paths.}
\label{fig:motivation}
\end{figure}

\F~\ref{fig:motivation}(a) shows a multi-layer fragment produced by a typical ZK
machine learning framework (we take ezkl in this example). Each gadget emits its
own soundness checks without consulting its neighbours, producing the
constraints listed in the right column:
\begin{itemize}[leftmargin=*]
  \item $c_{\mathsf{fc}}$: fully-connected gate enforcing $y = Wx + b$.
  \item $c_{\mathsf{relu}}$: ReLU lookup asserting $z \in [0, R_{\mathsf{relu}}]$.
  \item $c_{\mathsf{dec}}$: decomposition gate $z = s \cdot (d_1 B + d_0)$.
  \item $c_{d_0}, c_{d_1}$: digit lookups asserting $d_0, d_1 \in [0, B)$.
  \item $c_{\mathsf{sign}}$: sign lookup asserting
    $s \in \{-1, +1\}$.
  \item $c_{\mathsf{rng}}$: next-layer range check asserting $z \in [0,
    R_{\mathsf{rng}}]$.
\end{itemize}
In isolation every constraint is essential, but in the full circuit
redundancies emerge:
\begin{itemize}[leftmargin=*]
  \item \emph{Range-check subsumption.} Assume that $R_{\mathsf{rng}} \le
    R_{\mathsf{relu}}$. $c_{\mathsf{relu}}$ asserts $z \in [0,
    R_{\mathsf{relu}}]$, but the downstream $c_{\mathsf{rng}}$ already enforces
    $z \in [0, R_{\mathsf{rng}}]$, making $c_{\mathsf{relu}}$ redundant (dashed
    arrow in the figure).
  \item \emph{Bound-induced collapse.} Assume the sign and low-digit
    checks remain, and another surviving constraint bounds
    $m=d_1B+d_0$ to $[0,B)$. When $R_{\mathsf{rng}} < B$, the downstream bound
    gives $z \in [0,B)$, while $m \in [0,B)$ and $d_0 \in [0,B)$ force
    $d_1=0$ because $d_1B=m-d_0 \in (-B,B)$ is a multiple of $B$. Thus
    $c_{d_1}$ is redundant; $c_{\mathsf{sign}}$ can be removed only if
    another surviving fact pins $s=1$.
\end{itemize}
Identifying and safely removing such redundancies is the goal of this
work, but doing so raises three challenges.

\subsection{Design Challenges}
\label{sec:obstacles}


\parh{Challenge~A: Non-local subsumption.} The ReLU lookup
$c_{\mathsf{relu}}$ asserts $z \in [0, R_{\mathsf{relu}}]$ but is fully
subsumed by the downstream range check $c_{\mathsf{rng}}: z \in [0,
R_{\mathsf{rng}}]$ with $R_{\mathsf{rng}} \le R_{\mathsf{relu}}$; the
bound-induced collapse is similarly non-local. Certifying either redundancy
requires bounding $z$ from constraints spread across the whole circuit.\looseness=-1

Bounding variables from surrounding constraints is the canonical problem of
\emph{deep neural network (DNN)
verification}~\cite{katz2017reluplex,zhang2018efficient,singh2019abstract}. Yet
mature DNN techniques (linear bound propagation, SMT/MILP) assume a layered DAG
seeded at a designated input region or output property, which a compiled ZK
constraint system lacks. The bound on $z$, for instance, flows backward from
$c_{\mathsf{rng}}$ to $c_{\mathsf{relu}}$, and elsewhere bounds flow forward
across one equality and backward across the next. \emph{This rules out any
analysis tied to a fixed traversal direction.} We adopt abstract interpretation
in a worklist form (\S\ref{sec:abstract-interp}), whose fixpoint propagates
facts in any direction over the flat constraint set without a designated
starting point.\looseness=-1

\parh{Challenge~B: Premise asymmetry.} Consider \F~\ref{fig:motivation}(b): a
gate $c_1: x + y = 10$ and two lookups $c_2: x \in [0, 10]$,
$c_3: y \in [0, 10]$. The fact $x \in [0, 10]$ has two derivations: directly
from $c_2$, or jointly from $\{c_1, c_3\}$ since $x = 10 - y$ and $y \in [0,
10]$ force $x \in [0, 10]$; $y$'s derivations are symmetric. Each lookup is
individually redundant, yet removing both leaves only $c_1$, which admits the
unexpected witness $x = -5, y = 15$: each joint derivation required the very
lookup just deleted. Such mutual dependencies recur in ZK machine learning
circuits, e.g., paired bounds in division gadgets, cross-layer range checks.\looseness=-1

Formally, each $c_i$'s redundancy is established against $\mathcal{C} \setminus
\{c_i\}$, which still contains its partner. Safe debloating requires entailment
from the strictly smaller $\mathcal{C} \setminus \mathcal{R}$
(\Def~\ref{def:debloating}), in which every removed constraint is gone
simultaneously. We call this gap \emph{premise asymmetry}: a constraint
redundant given its peers may stop being so once those peers are also removed.
\emph{Safe composition therefore requires a data structure that tracks
alternate derivations of each fact, plus a search procedure that rejects
circular derivations among them.} We build a provenance graph
(\S\ref{sec:provenance}) recording which constraints proved which fact,
exposing alternate derivations to the removal loop.


\parh{Challenge~C: Intractable exact methods.} The decision
$\mathcal{C} \setminus \mathcal{R} \vdash \mathcal{R}$ is co-NP-hard already
for a single constraint (3-UNSAT reduces via a degree-3 polynomial encoding
over $\mathbb{F}_p$); SMT-based analyses do not scale to constraint systems of
this size~\cite{picus}, while the circuits we evaluate reach tens of
millions. The optimisation compounds the hardness: it ranges over
$2^{|\mathcal{C}|}$ candidates, and Challenge~B shows subset validity does not
decompose into singleton checks, placing the exact problem in $\Sigma_2^P$.\looseness=-1

We therefore forgo exact methods for a polynomial-time, sound-by-construction
over-approximation. Three practical limits shape its design.
\textbf{One at a time}: with $2^{|\mathcal{C}|}$ subsets to enumerate, we grow
$\mathcal{R}$ greedily, accepting each addition only if safely removable given
prior commits; the cumulative mask is also the snapshot Challenge~B's
proof search consults, so one-at-a-time is what makes cycle-aware checking
possible. \textbf{Shared work}: rerunning the analysis after each tentative
removal would multiply its cost by $|\mathcal{C}|$, not scalable
(\S\ref{sec:eval}); we instead compute the derivations once and reuse them via
the provenance graph. \textbf{Targeted obligations}: since checking
even single-constraint redundancy is hard, a heuristic \emph{entailment filter}
(\S\ref{sec:debloating}) flags constraints whose enforced property the fixpoint
already entails and names the facts the proof search must re-derive.\looseness=-1

\S\ref{sec:approach} instantiates these three responses as a single pipeline.

\section{Approach}
\label{sec:approach}

For the check classes handled by our entailment rules, removing a
constraint $c$ from $\mathcal{C}$ is safe iff the facts $c$ enforces are still
derivable from $\mathcal{C} \setminus \{c\}$; our pipeline turns that
derivability test into a tractable computation in three phases:
\[
\mathcal{C}
\;\xrightarrow[\,\S\ref{sec:abstract-interp}\,]{\text{\small abstract
interp.}}\; \sigma^\star
\;\xrightarrow[\,\S\ref{sec:provenance}\,]{\text{\small provenance}}\; G
\;\xrightarrow[\,\S\ref{sec:debloating}\,]{\text{\small debloating}}\;
\mathcal{R}. \] \emph{Abstract interpretation} proves which facts
hold of each variable, computing a fixpoint $\sigma^\star$ over three
abstract domains (intervals, known bits, constants). The \emph{provenance
graph} $G$ records, for each fact in $\sigma^\star$, which constraints
proved it from which premise facts. \emph{Provenance-guided debloating}
then verifies, for each candidate $c$, that the facts $c$ enforces remain
derivable in $G$ after $c$ is masked out; if so, $c$ is removed, yielding
the removal set $\mathcal{R}$. \reved{In
Appendix~\ref{sec:walkthrough} we give an end-to-end example of the three
phases on a five-constraint circuit distilled from
\S\ref{sec:motivating-example}.}

\subsection{Abstract Interpretation over Constraints}
\label{sec:abstract-interp}

\parh{Abstract domain.}
Each variable $v \in \mathcal{V}$ is tracked by an abstract value drawn from
a \emph{reduced product} of three lattice domains:

\begin{enumerate}[leftmargin=*]
  \item \emph{Constant}: either a known field element or $\top$ (unknown).
  \item \emph{Interval}: a range $[\ell, u] \subseteq \mathbb{Z}$ bounding all
    possible values of $v$ when interpreted as a signed integer.
  \item \emph{Known bits}: a bitmask $m$ and a value pattern $b$, recording
    that for every bit position $i$ where $m_i = 1$, the $i$-th bit of $v$
    equals $b_i$.
\end{enumerate}

The three components cross-propagate: a refinement in one domain sharpens the
others via standard reduced-product rules~\cite{cousot77, cousot79}, as in
LLVM's \texttt{KnownBits} analysis~\cite{llvm-knownbits}.

\parh{Transfer functions.}
Each constraint type defines a \emph{transfer function} that refines
the abstract values of its variables. Transfer functions handle arithmetic
gates (addition, multiplication, scaling, and their inverses), lookup tables
(intersecting a variable's abstract value with the table's range), and bit
decompositions (linking interval bounds to individual bit facts).
Propagation is non-directional: $x + y = z$ may tighten any of $x$, $y$, or
$z$ from the other two.\looseness=-1

We show three representative rules in inference-rule form; the remainder follow
the same template. Premises above the line are the abstract-state facts
consulted; the conclusion below the line is the contributed refinement. $\sigma$
denotes the current state; $A, B$ abbreviate $\sigma(a), \sigma(b)$;
$\sqsubseteq$ is the abstract-interval order (sharper bounds are smaller); and
$A \otimes B$ is the smallest interval containing $\{ab : a \in A,\ b \in B\}$
(interval multiplication).

{\small
\[
\begin{array}{c}\textsc{Lookup}\\[2pt]
\dfrac{v \in T \quad \mathrm{range}(T) = [t_\ell, t_u]}
      {[t_\ell, t_u] \sqsubseteq \sigma(v)}
\end{array}
\qquad
\begin{array}{c}\textsc{Bitify}\\[2pt]
\dfrac{\begin{array}{c}
        v = \textstyle\sum_{i<n} b_i \cdot 2^i \quad b_i \in \{0,1\} \\[2pt]
        \sigma(b_i) = \{0\} \text{ for } i \geq k
       \end{array}}
      {[0,\, 2^k - 1] \sqsubseteq \sigma(v)}
\end{array}
\]
\[
\textsc{Quadratic gate}\quad
\dfrac{a \cdot b + c = 0 \quad \sigma(a) = A \quad \sigma(b) = B}
      {-(A \otimes B) \sqsubseteq \sigma(c)}
\]
}

Symmetric variants are elided: \textsc{Bitify} also propagates backward, and
\textsc{Quadratic gate} solves backward for $a$ or $b$: \reved{rearranging
$a \cdot b + c = 0$ divides by the other operand, giving
$-(\sigma(c) / \sigma(b)) \sqsubseteq \sigma(a)$ and symmetrically for
$b$. The backward rule fires only when the divisor's range is
sign-determined (in particular, excludes zero)}, which makes the interval
division well defined; otherwise it contributes nothing.

\parh{Soundness under field wrap-around.}
Witnesses live in $\Fp$, not $\mathbb{Z}$. We interpret a field element
as its canonical signed representative in the window
$W = [-(p-1)/2,\,(p-1)/2]$. A transfer rule is allowed to contribute an
integer fact only when the concrete operation it abstracts is unique in this
signed interpretation: for addition and scaling, the operand ranges and the
resulting range must stay inside $W$; for multiplication, the full interval
product must stay inside $W$; and for backward rules such as solving a
quadratic gate for one operand, the divisor/range side conditions used by the
rule must also keep the inferred range inside $W$. Lookup and bit-decomposition
rules are treated as integer facts only for table entries or bit widths whose
represented values lie inside the same signed window. We enforce this
discipline by widening any out-of-window conclusion to $\top$:

\[
\textsc{WrapGuard}\quad
\dfrac{[\ell, u] \sqsubseteq \sigma(v) \;\text{(derived above)}
       \quad [\ell, u] \not\subseteq W}
      {\top \sqsubseteq \sigma(v)}
\]

Thus, any modular wrap-around that would make the integer reading
ambiguous causes the analysis to keep no refinement from that rule. With the
wrap-around guard, the analysis performs a refinement only when integer and
modular arithmetic agree on it, so every kept fact holds for every witness
(wrap-around or not).

\parh{Fixpoint computation.}
Variables start at $\top$, except those pinned to circuit constants (e.g.,
fixed-column entries or hard-coded values), which start at the corresponding
constant. A standard worklist re-applies each constraint's transfer function
whenever one of its variables is refined.

\parh{Soundness, termination, and complexity.}
Each transfer function is \emph{locally sound}: for every constraint $c$ and
every witness $w \models c$, applying $\mathrm{transfer}_c$ to any abstract
state that over-approximates $w$ yields a state that still over-approximates
$w$. We verify this condition per constraint type and take it as a
precondition on the rules. Local soundness, monotonicity
($\sigma_1 \sqsubseteq \sigma_2 \Rightarrow \mathrm{transfer}_c(\sigma_1)
\sqsubseteq \mathrm{transfer}_c(\sigma_2)$), and the bounded height of the
lattice (every interval lies in $W$ by \textsc{WrapGuard}) are the standard
ingredients of
the worklist fixpoint construction~\cite{cousot77, cousot79}: chaotic
iteration (a worklist with arbitrary dequeue order) reaches a fixpoint
$\sigma^\star$ that over-approximates every $w \models \mathcal{C}$.
With $n = |\mathcal{C}|$ constraints, $d$ the maximum constraint arity, and
$h$ the height of the interval lattice over $W$, the construction runs in
$O(n \cdot d \cdot h)$ time.

\subsection{Provenance Graph}
\label{sec:provenance}

The fixpoint $\sigma^\star$ proves facts but discards how they were
derived; the premise asymmetry of Challenge~B shows why that history matters.
To handle it without re-running the fixpoint after every candidate removal,
we record the derivation structure as a \emph{provenance
graph}~\cite{DOYLE1979231, assumptionbasedtms}: a bipartite AND/OR graph
whose fact nodes (OR) hold proven facts and whose constraint nodes (AND)
record which constraints, with which premise facts, established them.\looseness=-1

To track derivations at the granularity downstream entailment needs, we
\reved{\emph{project} each variable's reduced-product value $\sigma^\star(v)$ into
up to three \emph{atomic facts}, one per abstract domain: a lower bound, an
upper bound, and a known-bits pattern}. Each atom is separately provable, so
a constraint can take credit for the specific atom it contributes and
provenance tracks each atom on its own.

For a state $\sigma$ and a fact $k$ over variable $v$, we write $\sigma(v)
\Vdash k$ to mean that every concrete value in $\gamma(\sigma(v))$
satisfies~$k$. We use $\Vdash$ to distinguish abstract from semantic entailment
(\S\ref{sec:problem-formulation}).

\begin{definition}[Provenance graph]
\label{def:provgraph}
The \emph{provenance graph} of constraint system $\mathcal{C}$ at fixpoint
$\sigma^\star$ is a directed bipartite graph
$G = (F \cup N,\, E)$ with parts $F$ (fact nodes, the OR side) and
$N$ (constraint nodes, the AND side):
\begin{itemize}[leftmargin=*]
  \item $F \subseteq \mathcal{V} \times \mathcal{K}$ is the set of
    \emph{fact nodes}, one per pair $(v, k)$ with $k \in \mathcal{K} =
    \{\mathit{lo}, \mathit{hi}, \mathit{bits}\}$ and
    $\sigma^\star(v) \Vdash (v, k)$. Each node asserts the corresponding
    atom of $\sigma^\star(v)$: a lower bound $v \geq \ell$, an upper bound
    $v \leq u$, or a known-bits pattern (a bitmask plus value pinning
    specific bit positions of $v$), using the values $\sigma^\star(v)$
    records. This per-domain projection lets each entailment check
    consume only what it needs: bound checks read $\mathit{lo}/\mathit{hi}$,
    bit-level checks read $\mathit{bits}$.
  \item $N$ contains one \emph{constraint node} per triple $(c, f, P)$
    such that constraint $c \in \mathcal{C}$ derives target fact
    $f \in F$ from premise set $P \subseteq F$. A constraint may
    contribute several nodes (one per derivable fact). A node with
    $P = \emptyset$ is an \emph{axiom node}: $c$ derives $f$ from its
    semantics alone, independent of any other proven fact. Axiom
    nodes (e.g., a lookup $v \in T$ with column range $[t_\ell, t_u]$
    contributing the bound facts of $v$) cannot participate in dependency
    cycles, providing cycle-free proof paths the debloater can rely on.
    Constraint nodes with $P \neq \emptyset$ are \emph{joint
    nodes}, deriving their target from the conjunction of premise facts.
  \item $E$ contains, for each $\nu = (c, f, P) \in N$, a
    \emph{premise edge} from each $p \in P$ to $\nu$, and a
    \emph{derivation edge} from $\nu$ to its target $f$.
\end{itemize}
\end{definition}

\parh{Construction.}
Each constraint node $\nu = (c, (v, k), P)$ records what the fixpoint
does not track: which fact $(v, k)$ the constraint $c$ alone contributes, and
which premises $P$ that contribution consults. \Alg~\ref{alg:build-prov-graph}
computes this in one outer loop over $\mathcal{C}$; no fixpoint iteration is
needed. For each $(c, v)$, it \emph{replays} $c$'s transfer function with $v$
reset to $\top$ and other variables held at $\sigma^\star$; for the
monotone, single-constraint transfer functions implemented by our analyzer,
this replay exposes the refinements that $c$ can contribute to $v$ at the
fixpoint. We do not require this replay to be complete for all possible
semantic consequences of $c$; facts not exposed by an implemented transfer rule
are simply unavailable to the debloater. $\textsc{MinPremises}$ (below) then
trims $P$ to the facts the replay actually consulted.

\begin{figure}[t]
\begin{flushleft}
\algcaption{Provenance graph construction.}
\label{alg:build-prov-graph}
\hrule\smallskip
{\small
\Require Constraints $\mathcal{C}$, fixpoint $\sigma^\star$
\Ensure Provenance graph $G = (F \cup N, E)$
\State $F \gets \{(v, k) : v \in \mathcal{V},\ k \in \mathcal{K},\ \sigma^\star(v) \Vdash (v, k)\}$
\State $N \gets \emptyset$
\State \textbf{for} each $c \in \mathcal{C}$ \textbf{do}
\State \quad \textbf{for} each variable $v$ referenced by $c$ \textbf{do}
\State\label{alg:reset} \quad\quad $\sigma' \gets \sigma^\star[v \mapsto \top]$\hfill$\triangleright$ reset $v$
\State\label{alg:apply} \quad\quad $\sigma'' \gets \mathrm{transfer}_c(\sigma')$\hfill$\triangleright$ replay
\State\label{alg:loop} \quad\quad \textbf{for} each $k \in \mathcal{K}$ with $\sigma''(v) \Vdash (v, k)$ \textbf{do}
\State \quad\quad\quad $P \gets \textsc{MinPremises}(c, v, k, \sigma^\star)$
\State\label{alg:add-node} \quad\quad\quad $N \gets N \cup \{(c,\ (v, k),\ P)\}$
\State $E \gets$ premise and derivation edges induced by $N$ (\Def~\ref{def:provgraph})
\State \Return $(F \cup N,\ E)$
}
\smallskip\hrule
\end{flushleft}
\Description{Algorithm for constructing the provenance graph. It iterates over
constraints and referenced variables, resets one variable to top, replays the
constraint transfer, minimizes the premises of every derived fact, adds the
resulting constraint nodes, and returns the graph with premise and derivation
edges.}
\end{figure}

\parh{Minimal premises.}
The replay above held \emph{every} other-variable fact at $\sigma^\star$, but
most were not actually consulted. Recording the maximal set would inflate the
graph and, worse, manufacture spurious dependencies that the debloater would
later have to mask around. $\textsc{MinPremises}$ (\Alg~\ref{alg:min-premises}) trims them
via a \emph{drop-to-top} rule: candidates are individual facts $(v', k')$
about variables $v'$ that $c$ references, and the procedure tentatively
drops each in turn (widening the affected variable toward $\top$ via
$\mathit{state}(\cdot)$ from \Alg~\ref{alg:min-premises}), replays
$\mathrm{transfer}_c$, and discards a candidate if the target still
emerges. Soundness follows from
\S\ref{sec:abstract-interp}'s monotonicity: a fact whose absence the target
survives was not consulted. Visit
order is arbitrary; different orders yield different irreducible sets, but
proof search needs \emph{some} proof, not a minimum-cardinality one.\looseness=-1

\begin{figure}[t]
\begin{flushleft}
\algcaption{$\textsc{MinPremises}$ via drop-to-top.}
\label{alg:min-premises}
\hrule\smallskip
{\small
\par\noindent\textit{Notation.} For a fact set $S \subseteq F$, let
$\mathit{state}(S)$ denote the abstract state assigning $v \mapsto \top$
and each other variable $u$ the narrowest abstract value consistent
with the facts $\{(u, k') \in S\}$.
\par\smallskip
\Require Constraint $c$, target fact $(v, k)$, fixpoint $\sigma^\star$
\Ensure Irreducible $P \subseteq F$ such that
$\mathrm{transfer}_c(\mathit{state}(P))(v) \Vdash (v, k)$
\State $P \gets \{(v', k') : v' \in \mathrm{vars}(c) \setminus \{v\},\ k' \in \mathcal{K},\ \sigma^\star(v') \Vdash (v', k')\}$
\State \textbf{for} each $p \in P$ (arbitrary order) \textbf{do}
\State \quad $\sigma' \gets \mathit{state}(P \setminus \{p\})$\hfill$\triangleright$ widen by dropping $p$
\State \quad \textbf{if} $\mathrm{transfer}_c(\sigma')(v) \Vdash (v, k)$ \textbf{then}
\State \quad\quad $P \gets P \setminus \{p\}$\hfill$\triangleright$ $p$ unused
\State \Return $P$
}
\smallskip\hrule
\end{flushleft}
\Description{Algorithm for minimizing the premises of a derived fact. Starting
from all relevant fixpoint facts, it drops each premise in turn, reconstructs
the abstract state, and permanently removes the premise when replaying the
constraint still derives the target fact.}
\end{figure}

\parh{Complexity.}
Write $n = |\mathcal{C}|$ and $d$ for the maximum number of variables referenced
by any single constraint.

\begin{theorem}[Construction complexity]
\label{thm:prov-complexity}
\Alg~\ref{alg:build-prov-graph} runs in time $O(n \cdot d^3)$ and produces
a graph with $O(|\mathcal{V}|)$ fact nodes, $O(n \cdot d)$ constraint nodes, and
$O(n \cdot d^2)$ edges.
\end{theorem}

Usually, a constraint has a limited number of variables, so $d$ is a small
constant, thus the construction is near-linear in $n$.
\reved{The proof sketch appears in Appendix~\ref{sec:appendix-proofs}.}\looseness=-1

\subsection{Provenance-Guided Debloating}
\label{sec:debloating}

Debloating selects constraints whose removal preserves the witness set.
An \emph{entailment filter} proposes candidates; a \emph{removal loop}
removes them one at a time, using AND/OR \emph{proof search} on $G$ to
certify each removal. We describe these components in order, then argue
soundness.

\parh{Entailment filter.}
A constraint $c$ is a \emph{candidate}, written $c \in
\mathcal{C}_{\mathsf{f}}$, when the fixpoint $\sigma^\star$ already implies
its enforced property: $\sigma^\star \Vdash_{\mathsf{ent}} c$. We define
$\Vdash_{\mathsf{ent}}$ by cases on the constraint shape; two are
representative, and the remainder follow the same template.

\begin{itemize}[leftmargin=*]
  \item \textsc{Lookup} $c:\ v \in T$ where $T$ is a single-column range
    table containing exactly the integers in $[t_\ell, t_u]$ — holds when
    $\sigma^\star(v) \Vdash v \in [t_\ell, t_u]$.
  \item \textsc{Gate} $c:\ e(\vec v) = 0$ — holds when
    $\sigma^\star(v_i) = \{a_i\}$ for all $i$ and $e(\vec a) = 0$.
\end{itemize}

Here $\sigma^\star(v) \Vdash v \in [t_\ell, t_u]$ abbreviates the joint
$\mathit{lo}/\mathit{hi}$ atom check, and $\sigma^\star(v_i) = \{a_i\}$ denotes
a singleton interval (lo and hi coincide at $a_i$). For lookup tables that are
not of this form, the range rule does not apply. Analogous cases
handle bit decompositions, R1CS products, and equals-zero constraints. The
filter is sound but incomplete: a redundant constraint whose shape does not
match any implemented case is simply retained.

\parh{Provability under a mask.}
A subset $M \subseteq N$ is a \emph{mask}, modeling constraint nodes that the
debloater has tentatively disabled; masks arise because removals are applied
incrementally and provability must be re-checked after each. Provability
under $M$ has AND/OR semantics: a fact $f \in F$ is provable iff some
unmasked constraint node $\nu \in N \setminus M$ targeting $f$ is itself
provable (OR), and a constraint node $\nu = (c, f, P) \in N \setminus M$ is
provable iff every premise $p \in P$ is provable (AND). We take the least
fixed point of these clauses: a fact is provable only when a finite, acyclic
derivation rooted at axiom nodes witnesses it; cyclic derivations do not
establish provability.\looseness=-1

\parh{Proof search.}
The removal loop needs to decide, for a fact $f$ and a mask $M$,
whether $f$ is provable in $G$ under $M$. We compute this by AND/OR
depth-first search (\Alg~\ref{alg:proof-search}).
$\textsc{ProveOr}(f, M)$ succeeds when some unmasked constraint node $\nu$
targeting $f$ succeeds; a constraint node $\nu \in N \setminus M$ succeeds when
every premise of $\nu$ does. The $\mathit{visiting}$ set returns \textsc{false}
on revisits to a node still on the search stack, refusing to derive a fact from
a circular proof and thereby ruling out the premise-asymmetry cycles
of Challenge~B (\S\ref{sec:obstacles}). The $\mathit{proven}$ and
$\mathit{disproven}$ memos record each fact node's result on return, so subsequent
references within the same query short-circuit. Each node and edge is thus
traversed at most once per query, giving $O(|F| + |N| + |E|)$ time.

\begin{figure}[t]
\begin{flushleft}
\algcaption{AND/OR proof search under a mask.}
\label{alg:proof-search}
\hrule\smallskip
{\small
\Require Graph $G = (F \cup N, E)$, target $f \in F$, mask $M \subseteq N$
\Ensure \textsc{true} iff $f$ is provable in $G$ under $M$
\par\smallskip\noindent\textit{Driver:}
\State $\mathit{visiting},\, \mathit{proven},\, \mathit{disproven} \gets \emptyset$
\State \Return $\textsc{ProveOr}(f, M)$
\par\smallskip\noindent\textit{function} $\textsc{ProveOr}(f, M)$:
\State \quad \textbf{if} $f \in \mathit{proven}$ \textbf{return} \textsc{true}\hfill$\triangleright$ memo
\State \quad \textbf{if} $f \in \mathit{disproven}$ \textbf{return} \textsc{false}\hfill$\triangleright$ memo
\State \quad \textbf{if} $f \in \mathit{visiting}$ \textbf{return} \textsc{false}\hfill$\triangleright$ cycle
\State \quad $\mathit{visiting} \gets \mathit{visiting} \cup \{f\}$;\quad $r \gets \textsc{false}$
\State \quad \textbf{for} each $\nu$ with $(\nu \to f) \in E$ \textbf{do}
\State \quad\quad \textbf{if} $\textsc{ProveAnd}(\nu, M)$ \textbf{then} $r \gets \textsc{true}$; \textbf{break}
\State \quad $\mathit{visiting} \gets \mathit{visiting} \setminus \{f\}$
\State \quad \textbf{if} $r$ \textbf{then} $\mathit{proven} \gets \mathit{proven} \cup \{f\}$ \textbf{else} $\mathit{disproven} \gets \mathit{disproven} \cup \{f\}$
\State \quad \textbf{return} $r$
\par\smallskip\noindent\textit{function} $\textsc{ProveAnd}(\nu = (c, f', P), M)$:
\State \quad \textbf{if} $\nu \in M$ \textbf{return} \textsc{false}
\State \quad \textbf{if} $\nu \in \mathit{visiting}$ \textbf{return} \textsc{false}\hfill$\triangleright$ cycle
\State \quad $\mathit{visiting} \gets \mathit{visiting} \cup \{\nu\}$;\quad $r \gets \textsc{true}$
\State \quad \textbf{for} each $p \in P$ \textbf{do}
\State \quad\quad \textbf{if} $\neg \textsc{ProveOr}(p, M)$ \textbf{then} $r \gets \textsc{false}$; \textbf{break}
\State \quad $\mathit{visiting} \gets \mathit{visiting} \setminus \{\nu\}$;\ \textbf{return} $r$
}
\smallskip\hrule
\end{flushleft}
\Description{Cycle-aware AND-OR proof-search algorithm under a mask. The driver
initializes visiting, proven, and disproven sets. The OR procedure searches for
an unmasked supporting constraint, while the AND procedure recursively proves
every premise; stack revisits fail to prevent circular proofs.}
\end{figure}

\parh{Removal loop.}
For each candidate $c \in \mathcal{C}_{\mathsf{f}}$, its
\emph{proof obligations} $\Obl(c) \subseteq F$ are the fact nodes the filter
consulted when admitting $c$. The two filter rules above instantiate
$\Obl(c)$ as follows: a lookup entailed by
$\sigma^\star(v) \Vdash v \in [t_\ell, t_u]$ contributes the $\mathit{lo}$
and $\mathit{hi}$ nodes of $v$; a gate entailed by $\sigma^\star(v_i) =
\{a_i\}$ for all $i$ contributes the $\mathit{lo}$ and $\mathit{hi}$ nodes
of every $v_i$ (which together pin $v_i$ to $a_i$). If every fact in $\Obl(c)$
remains provable under the cumulative mask of prior removals, the
loop removes $c$ and adds its constraint nodes to the mask
(\Alg~\ref{alg:commit-loop}).

\begin{figure}[t]
\begin{flushleft}
\algcaption{Removal loop for provenance-grounded debloating.}
\label{alg:commit-loop}
\hrule\smallskip
{\small
\Require Candidates $\mathcal{C}_{\mathsf{f}}$, provenance graph
$G = (F \cup N, E)$, ordering $\pi$ on $\mathcal{C}_{\mathsf{f}}$
\Ensure Removal set $\mathcal{R}$
\State $\mathcal{R} \gets \emptyset$;\quad $M \gets \emptyset$
\State \textbf{for} each $c \in \mathcal{C}_{\mathsf{f}}$ in order $\pi$ \textbf{do}
\State \quad $M_c \gets \{\,\nu = (c', \cdot, \cdot) \in N : c' = c\,\}$\hfill$\triangleright$ $c$'s constraint nodes
\State \quad $M \gets M \cup M_c$\hfill$\triangleright$ tentatively mask $c$
\State \quad \textbf{if} $\textsc{ProveOr}(f, M)$ succeeds for every $f \in \Obl(c)$ \textbf{then}
\State \quad\quad $\mathcal{R} \gets \mathcal{R} \cup \{c\}$\hfill$\triangleright$ remove; mask stays
\State \quad \textbf{else}
\State \quad\quad $M \gets M \setminus M_c$\hfill$\triangleright$ revert tentative mask
\State \Return $\mathcal{R}$
}
\smallskip\hrule
\end{flushleft}
\Description{Greedy removal-loop algorithm. For each candidate constraint, it
tentatively masks all graph nodes belonging to that constraint and proves every
removal obligation. It retains the mask and records the removal on success, or
restores the nodes on failure.}
\end{figure}

\parh{Ordering and intractability.}
Finding the largest removal set is at least as hard as the minimum equivalent
subformula problem, which is coNP-hard~\cite{Liberatore_2005}: the maximum removal
set corresponds to a smallest equivalent subset of $\mathcal{C}$, and deciding
whether such a subset of a given size exists is intractable. We therefore use a
greedy ordering (\S\ref{sec:impl}) that visits low-uniqueness candidates first,
removing constraints with many alternative proofs before those that are the sole
proof source for some fact. The ordering affects only the \emph{size} of the
removal set, not its \emph{soundness}: every removal is individually
verified by proof search.

\parh{Complexity.}
Let $m = |\mathcal{C}_{\mathsf{f}}|$ denote the number of candidates emitted by
the filter and $\bar D$ the average size of the subgraph reachable from a target
fact in $G$. Each removal attempt runs $|\Obl(c)|$ proof searches, each bounded
by $O(\bar D)$ under memoized cycle-aware traversal
(\Alg~\ref{alg:proof-search}). With $|\Obl(c)|$ bounded by the constant
constraint arity, one attempt costs $O(\bar D)$, and the entire loop runs in
$O(m \cdot \bar D)$ time. In the worst case $\bar D = O(m)$, giving $O(m^2)$; in
circuits where $\bar D$ is small, the loop is near-linear.

\medskip

For removals justified by the implemented transfer and entailment
rules, soundness proceeds in two steps. Theorem~\ref{thm:prov-soundness}
shows that any fact the search declares provable under a mask holds for every
witness of the unmasked constraints. Combined with a per-rule
entailment-soundness precondition on the filter, Theorem~\ref{thm:soundness}
then concludes that those removed constraints are redundant.

\parh{Provability soundness.}
Let $R \subseteq \mathcal{C}$ and let
$M_R = \{\,\nu = (c, f, P) \in N : c \in R\,\}$ be the mask induced
by removing $R$.

\begin{theorem}[Provability soundness]
\label{thm:prov-soundness}
If a fact $f = (v, k)$ is provable in $G$ under mask $M_R$, then every
witness $w \models \mathcal{C} \setminus R$ assigns $v$ a value
satisfying fact $k$.
\end{theorem}

\parh{Debloating soundness.}
Iterating \Alg~\ref{alg:commit-loop}'s removals in order produces
a \emph{provenance-grounded removal set}:

\begin{definition}[Provenance-grounded removal]
\label{def:grounded}
A subset $\mathcal{R} \subseteq \mathcal{C}$ is \emph{provenance-grounded} if
there exists an ordering $c_1, c_2, \ldots, c_k$ of its elements such that each
$c_i$ is entailed (so its obligations $\Obl(c_i)$ are defined) and, writing
$\mathcal{R}_i = \{c_1, \ldots, c_i\}$, every fact in $\Obl(c_i)$ is provable in
$G$ under mask $M_{\mathcal{R}_i}$ for every $i \in \{1, \ldots, k\}$.
\end{definition}

Each entailment rule is engineered to satisfy an \emph{entailment soundness}
property: for every entailed constraint $c$ and every witness $w$ that satisfies
every fact in $\Obl(c)$ at its proven abstract value, $w$ also satisfies $c$.
The theorem below is therefore conditional on the rule schemas covered
by the implementation and on their local soundness obligations; unsupported
constraint shapes are conservatively retained rather than certified by this
argument.

\begin{theorem}[Debloating soundness]
\label{thm:soundness}
If $\mathcal{R}$ is provenance-grounded then $\mathcal{R}$ is a safe
debloating (\Def~\ref{def:debloating}): \reved{$\mathcal{C} \setminus
\mathcal{R} \vdash \mathcal{R}$}, equivalently every $w$ satisfies
$\mathcal{C} \setminus \mathcal{R}$ iff it satisfies $\mathcal{C}$.
\end{theorem}

\reved{Proof sketches for both theorems appear in Appendix~\ref{sec:appendix-proofs}.}

\parh{Security implication.}
Theorem~\ref{thm:soundness} guarantees that the debloated circuit $\mathcal{C}
\setminus \mathcal{R}$ admits exactly the same set of valid witnesses as the
original $\mathcal{C}$ for the certified removal set $\mathcal{R}$, whose
members are discharged by the implemented entailment rules. In the adversarial
setting of a zero-knowledge proof system, this means a malicious prover who
could not forge a valid witness before debloating still cannot do so afterward:
the soundness guarantee of the proof system is preserved for the
certified removals under the modeled constraint semantics.\looseness=-1

\parh{Remark on stale proofs.}
$\Obl(c_i)$ need only be provable under the mask \emph{at step $i$}, not
under the final mask. A later removal may invalidate $c_i$'s original proof
without invalidating its removal. The removal loop therefore needs no
bookkeeping for prior proofs: checking each obligation once against the
cumulative mask suffices.

\section{Implementation}
\label{sec:impl}

We implement the approach of \S\ref{sec:approach} in Rust in roughly
$30\text{k}$ lines. This section reports the engineering details in our
implementation.

\parh{Loaders and emitters.}~The tool supports two constraint formats.
\begin{itemize}[leftmargin=*]
  \item \textbf{R1CS.} The loader consumes Circom source via an
    ANTLR4~\cite{antlr4}-generated parser to produce a flat constraint vector.
    The emitter re-generates Circom source so that existing tool-chains can
    prove and verify the debloated circuit unchanged.
  \item \textbf{PLONKish.} halo2 circuits are built programmatically inside Rust
  rather than from a source file, so there is nothing to parse. The loader
  instead observes halo2's mock prover during a single synthesis pass and
  records every gate, lookup, and copy constraint into a JSON description that
  preserves the circuit structure. The emitter rebuilds the JSON back into a
  halo2 circuit, so the existing prover and verifier run unchanged.
\end{itemize}
Notably, when loading, we assume that model weights (encoded in PLONKish's
advice columns) are not known to our tool. We therefore do not leverage any
weight-specific optimizations.


\parh{Greedy removal ordering.}~\Alg~\ref{alg:commit-loop}'s ordering $\pi$
is realized as a \emph{uniqueness score}: for each candidate $c$, we count the
fact nodes for which $c$ is the sole unmasked source in the provenance graph.
Candidates are visited in \emph{ascending} score order, so constraints with
score~$0$ (those whose facts have alternative proofs) are tried first, while
high-score constraints that prop up many facts uniquely are deferred to the end.
This is the standard greedy surrogate for set cover, computed once from the
initial graph.

\parh{Dead elimination.}~After debloating, some constraints may reference only
auxiliary variables that are disconnected from the circuit's designated boundary
variables (e.g., public and private inputs/outputs). We apply a conservative
cleanup pass only to such unreachable auxiliary structure: a
backward-reachability pass starts from these boundary variables, marks the
constraints that can affect them as live, and treats the deletion of the rest as
preserving the observable relation. A subsequent column-pruning pass identifies
advice and fixed columns that no surviving constraint references and marks them
for removal.

\parh{Mechanized soundness proof.}~We mechanize the soundness argument
for certified removals in Rocq~\cite{Rocq}. The development defines the modeled
constraint semantics and abstract domains, proves local soundness for the
implemented transfer and entailment rules, and independently checks the
analyzer-emitted certificate for each removal set $\mathcal{R}$, \reved{establishing
the safe-debloating property $\mathcal{C} \setminus \mathcal{R} \vdash
\mathcal{R}$ (\Def~\ref{def:debloating}): every witness of the debloated
circuit satisfies the removed checks, so both circuits accept the same
witnesses}. This certification
covers the constraint kinds discharged by our implementation; unsupported
shapes are retained unless a covered rule certifies their removal.


\section{Evaluation}
\label{sec:eval}

\begin{table*}[t]
\centering
\caption{RQ1: Constraint reduction and proof-system impact on named benchmarks.
{$|C|$ is the number of loaded constraints; Red\% is the fraction removed by
debloating; $T_P$/$T_V$ are baseline prove/verify time; $|\pi|$ is the proof
size; $\Delta P$/$\Delta V$/$\Delta|\pi|$ are the relative changes in
prove/verify time, and proof size after debloating.}}
\label{tab:rq1}
\small
\setlength{\tabcolsep}{3pt}
\resizebox{\textwidth}{!}{
\begin{tabular}{l|rrrrrrrr|rrrrrrrr}
\toprule
& \multicolumn{8}{c|}{\textbf{zkml}} & \multicolumn{8}{c}{\textbf{ezkl}} \\
\cmidrule(lr){2-9} \cmidrule(lr){10-17}
Model & $|C|$ & Red\% & $T_P$ & $\Delta P$ & $T_V$ & $\Delta V$ & $|\pi|$ & $\Delta|\pi|$ & $|C|$ & Red\% & $T_P$ & $\Delta P$ & $T_V$ & $\Delta V$ & $|\pi|$ & $\Delta|\pi|$ \\
\midrule
MNIST-MLP & 6.30K & 22.30\% & 569\,s & $-$16.80\% & 13.70\,s & $-$25.30\% & 8.10\,MB & $-$2.70\% & 416.50K & 3.20\% & 174\,s & $-$20.70\% & 2.33\,s & 0.00\% & 7.90\,KB & $-$31.60\% \\
LeNet-5 & 338.10K & 31.70\% & 994\,s & $-$41.20\% & 9.82\,s & $-$59.10\% & 3.00\,MB & $-$40.00\% & 2.20M & 15.90\% & 298\,s & $-$32.00\% & 2.33\,s & +0.10\% & 15.20\,KB & $-$42.30\% \\
MobileNet-v2 & --- & --- & --- & --- & --- & --- & --- & --- & 25.30M & 22.80\% & 9545\,s & $-$4.10\% & 0.93\,s & $-$0.50\% & 275.00\,KB & $-$3.60\% \\
PTB-RNN & 28.10K & 43.10\% & 4.14\,s & $-$36.70\% & 32.00\,ms & $-$25.20\% & 58.20\,KB & $-$36.90\% & 3.80M & 12.90\% & 1561\,s & $-$62.90\% & 2.37\,s & $-$1.30\% & 93.80\,KB & $-$74.20\% \\
PTB-GRU & 78.80K & 44.70\% & 12.60\,s & $-$42.70\% & 94.80\,ms & $-$44.90\% & 178.20\,KB & $-$43.80\% & 12.80M & 12.10\% & 4.50\,h & $-$1.10\% & 1.06\,s & $-$6.60\% & 462.40\,KB & $-$1.20\% \\
PTB-LSTM & 76.10K & 48.70\% & 29.90\,s & $-$43.70\% & 0.16\,s & $-$48.10\% & 260.20\,KB & $-$45.60\% & 16.20M & 11.50\% & 7.60\,h & $-$20.10\% & 1.16\,s & +3.40\% & 730.50\,KB & $-$1.00\% \\
BERT-tiny & 57.20K & 37.40\% & 9.21\,s & $-$39.40\% & 71.00\,ms & $-$39.20\% & 137.00\,KB & $-$40.90\% & 10.20M & 24.90\% & 4783\,s & $-$72.80\% & 0.97\,s & $-$10.70\% & 138.10\,KB & $-$72.90\% \\
\bottomrule
\end{tabular}}
\end{table*}

We evaluate our approach along three dimensions:

\begin{itemize}[leftmargin=*]
  \item \textbf{RQ1 (Effectiveness).} How much does debloating reduce the number
    of constraints, and does that reduction translate into prover, verifier, and
    proof-size savings?
  \item \textbf{RQ2 (Scalability).} How does analysis time and constraint
    reduction scale as circuit size grows?
  \item \textbf{RQ3 (Comparison).} How does our tool compare to existing
    constraint-reduction tools and to exact SMT-based entailment checking?
\end{itemize}

\subsection{Experimental Setup}
\label{sec:eval-setup}

\parh{Hardware.}
All experiments run on a server equipped with an AMD EPYC~9754 and 768~GiB of
DDR5 RAM with Ubuntu~24~LTS.

\parh{Corpus.}
{\emergencystretch=1em
Our corpus covers seven architectures: MNIST-MLP (fully connected) and
LeNet-5 (convolutional) for digit recognition~\cite{mnist_and_lenet};
MobileNet-v2 (depthwise-separable convolutions)~\cite{mobilenetv2};
PTB-RNN (Elman)~\cite{rnn}, PTB-GRU~\cite{cho-etal-2014-learning}, and
PTB-LSTM~\cite{lstm} on Penn Treebank; and BERT-tiny (a single-layer
transformer encoder)~\cite{turc2019wellreadstudentslearnbetter}.

\noindent Each architecture is compiled through two PLONKish pipelines:
\texttt{ezkl}~\cite{ezkl}, which compiles ONNX models, and
\texttt{zkml}~\cite{zkml}, which composes hand-crafted PLONKish gadgets. Six of
the seven architectures run through both pipelines; MobileNet-v2 runs through
\texttt{ezkl} only, as \texttt{zkml} does not support depthwise-separable
convolutions. Running the same architecture through two pipelines separates
architecture-dependent redundancy (a property of the neural network model) from
pipeline-dependent redundancy (a property of the circuit generation framework).
Each model uses the published architecture dimensions where memory permits; when
the zkml prover's footprint exceeds the server's 768 GiB capacity, we use the
largest scale that completes successfully (hidden=50 for zkml/PTB-*,
$d_\text{model}=32$ for zkml/BERT-tiny). Circuit complexity at these sizes
ranges from $\sim\!6$k constraints to $\sim\!25$M constraints.\looseness=-1\par}

\parh{Proving and verifying.}
For end-to-end proof-system measurements in RQ1, we use Axiom's halo2
implementation~\cite{axiom-halo2} as it provides better memory efficiency,
and record CPU time for both proving and verifying.

\parh{Scaling.}
For RQ2, we parameterize each architecture family by a size multiplier
$\{1, 2, 4, 8, 16, 32\}\times$ that varies the principal dimension (hidden
width for recurrent/transformer models, channel count for CNNs, layer width
for FC networks), producing circuits that span the same structural patterns
at increasing scale. For each configuration we record our approach's analysis
wall time and the resulting constraint reduction rate. The largest instance
reaches $\sim\!38$M constraints at $32\times$.

\parh{Comparison baselines.}
For RQ3, we compare against five existing tools and one exact method. Circom
v2.2.3 with full optimization
(\texttt{O2})~\cite{belles2023circomacircuitdescription}, Distilling
Constraints~\cite{distillconstraints}, and CirC~\cite{Ozdemir2022CirCCI}
target R1CS constraint systems; since our tool also supports R1CS, the
comparison is direct. Circom~O2 and Distilling consume Circom source; CirC
compiles ZoKrates~\cite{zokrates} programs to R1CS, so we re-implement each
circuit as a semantically equivalent ZoKrates program.
\reved{Clap~\cite{stronati2024clapsemanticpreservingoptimizingedsl} and the Halo2
analyzer (korrekt)~\cite{halo2-analyzer} consume PLONKish circuits instead,
so the comparison uses the \texttt{zkml} compilation of the same
architectures: a converter reproduces each circuit's gates and lookups in
Clap's eDSL and runs Clap unmodified; korrekt requires Rust circuit sources
rather than extracted circuits, so we faithfully reimplement its three
structural analyses on the same circuits.} \reved{Finally, we
run the exact method our analysis approximates: deciding each redundancy
query $\mathcal{C} \setminus \{c\} \vdash c$ with cvc5~\cite{cvc5} on a
faithful QF\_NIA encoding. The encoding asserts gate polynomials as
equations, expands lookups into per-row disjunctions, and encodes copy
constraints as equalities. Each circuit receives a
24-hour budget; queries the solver cannot decide conservatively keep the
constraint.}\looseness=-1

The R1CS tools additionally need an R1CS corpus. To the best of our
knowledge, no mature R1CS framework exists for neural-network inference
circuits. The commonly
cited \texttt{circomlib-ml}~\cite{circomlib-ml} on Github is toy-level and
under-constrained: its gadgets omit checks needed for soundness. Thus, we
prepare five circuits (MLP, LeNet, RNN, GRU, LSTM) by augmenting
\texttt{circomlib-ml}'s templates with checks modeled after \texttt{zkml}'s
constraint patterns, producing sound circuits with realistic redundancy
structure. MobileNet and BERT are excluded because neither can be encoded using
gadgets provided by \texttt{circomlib-ml}. The sizes of those circuits are
intentionally small to allow all baselines to complete within a reasonable
time.\looseness=-1

The \emph{No-Graph} ablation removes each candidate, re-runs abstract
interpretation to fixpoint, and restores the candidate if entailment fails.
Each of the $|\mathcal{C}_{\mathsf{f}}|$ filter candidates triggers a fresh
$O(n \cdot d \cdot h)$ fixpoint: $O(n^2 \cdot d \cdot h)$ overall, versus
$O(n \cdot d^3)$ graph construction and an $O(m \cdot \bar D)$ commit loop
(\S\ref{sec:approach}). Its 4-hour budget per circuit is roughly $20\times$
our full approach's actual time on RQ3 circuits.\looseness=-1

\subsection{RQ1: Effectiveness}
\label{sec:eval-rq1}

Table~\ref{tab:rq1} reports constraint reduction, proving and verification
time, and proof size for each architecture--pipeline pair.

\begin{figure*}[t]
  \centering
  \includegraphics[width=\textwidth]{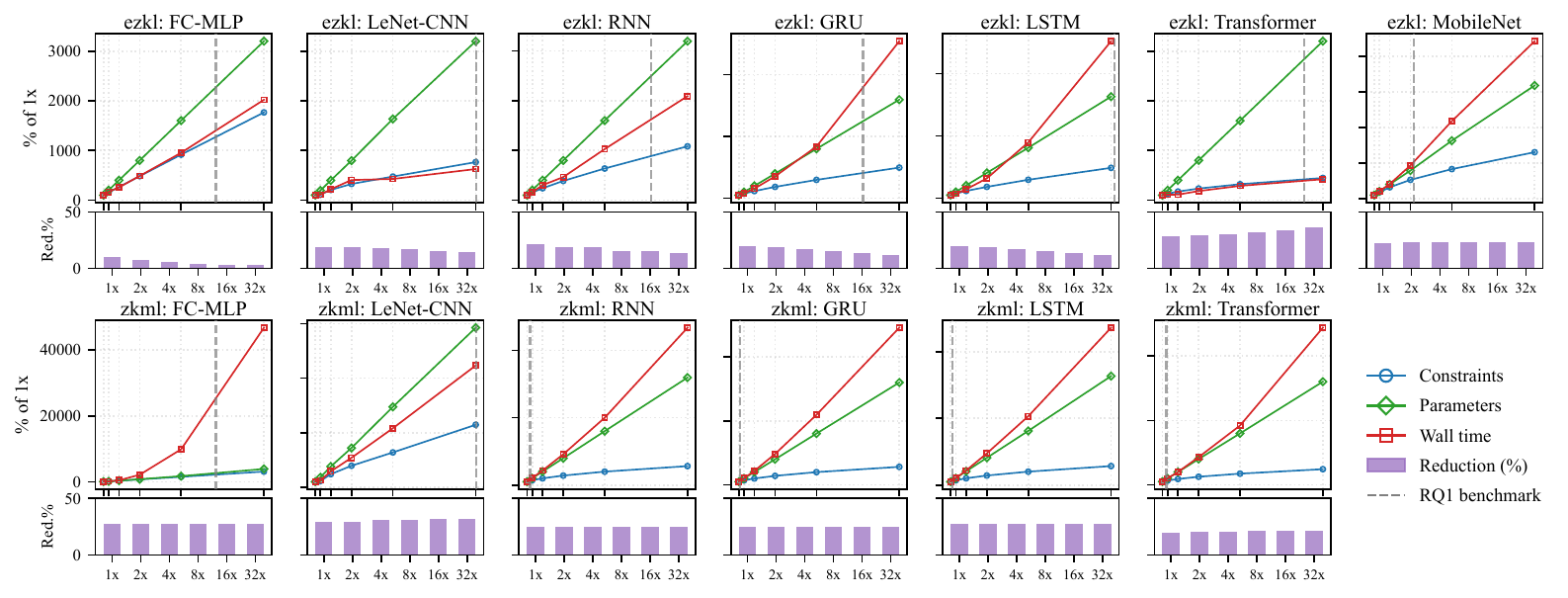}
  \caption{RQ2: Analysis time and constraint reduction as circuit size scales
    from $1\times$ to $32\times$. {Top row: \texttt{ezkl}; bottom row:
    \texttt{zkml}. In each panel, the upper chart plots constraints, parameters,
    and time relative to the $1\times$ baseline (log scale); the lower chart
    plots reduction percentage (linear scale). The dashed vertical line marks
    the parameter sizes of RQ1 models.}}
  \Description{Thirteen small-multiple plots arranged in two rows. Seven ezkl
    model families appear in the top row and six zkml families in the bottom
    row. Each panel shows constraints, parameters, and wall time relative to
    the one-times baseline from one-times through thirty-two-times scale, plus
    a lower bar chart of reduction percentage. Constraint counts and wall time
    increase with scale; reduction percentages remain broadly stable, with
    modest declines for several ezkl recurrent and fully connected models. A
    dashed vertical line marks the corresponding RQ1 benchmark size.}
  \label{fig:scaling}
\end{figure*}

\parh{Constraint reduction.}
Debloating removes constraints on every benchmark in our corpus. On
\texttt{ezkl}, reduction ranges from $3.20$\% to $24.90$\% (median $12.90$\%);
on \texttt{zkml}, from $22.30$\% to $48.70$\% (median $40.25$\%). The
consistently higher reduction on \texttt{zkml} is expected: \texttt{zkml}'s
hand-crafted gadgets compose range checks and lookup constraints at the
constraint level, producing inter-constraint entailment opportunities that
abstract interpretation can exploit. \texttt{ezkl}'s gadgets, by
contrast, fuse multiple operations into shared layouts where each soundness
check is emitted once and reused across many rows, leaving fewer
constraint-level redundancies for the debloater to detect. Within each
pipeline, recurrent and transformer architectures (PTB-*, BERT-tiny) yield
higher reduction than feed-forward ones (MNIST-MLP), consistent with unrolled
timesteps repeating identical gadget patterns whose intermediate checks
accumulate redundancy.\looseness=-1

\parh{Proving time.}
Prove-time savings appear on every benchmark, but the magnitude varies widely:
$16.80$--$43.70$\% on \texttt{zkml} and $1.10$--$72.80$\% on \texttt{ezkl}.
Several properties of PLONKish explain the spread.

The prover commits one polynomial per advice column, of length $N = 2^k$ where
$k$ (\texttt{logrows}) is the smallest power of two fitting the circuit; each
\emph{distinct} lookup argument adds further committed polynomials. The prover
therefore pays per advice column and per distinct lookup type, not per
constraint, so prove-time savings arise only when removals (i)~prune an entire
advice column, or (ii)~eliminate a distinct lookup argument. Witness-generation
also contributes to timing, scaling roughly with the number of cells in the
table (advice columns $\times$ rows used).\looseness=-1

On \texttt{zkml}, debloating predominantly removes lookup invocations
($33$--$50$\% fewer across the six benchmarks), while gate and advice-column
counts remain nearly unchanged. Because \texttt{zkml} emits a fresh lookup
argument per gadget call site, trimming invocations directly collapses the
\emph{number} of distinct lookup arguments. Prove-time tracks lookup reduction
closely ($17$--$44$\%, mean $36.75$\%).\looseness=-1

On \texttt{ezkl}, the dominant effect is column pruning. \texttt{ezkl}
shares lookup arguments across many call sites of the same non-linearity, so
distinct-lookup elimination is a smaller, secondary effect. On BERT-tiny and
PTB-RNN, debloating prunes $38$\% and $62$\% of advice columns and collapses
$29$\% and $17$\% of the distinct lookup arguments, yielding $72.80$\% and
$62.90$\% prove-time reductions. On the
remaining \texttt{ezkl} models, fewer columns are eliminated and speedups are
$1$--$32$\%.\looseness=-1

\parh{Verifier time and proof size.}
Verifier time is small in absolute terms and varies with system-level factors;
we observe $+3.40$\% to $-59.10$\% changes across benchmarks. Proof-size
reductions range from $1$--$74$\% on \texttt{ezkl} and $3$--$46$\% on
\texttt{zkml}, generally tracking the degree of column or lookup removal
discussed above.\looseness=-1

\begin{rqanswer}
\textbf{Answer to RQ1.}
Debloating finds and removes redundant constraints on all $13$ pipeline-model
pairs. The debloater reduces constraint counts by up to $48.70$\%, effectively
saving prove time and proof size.
\end{rqanswer}

\subsection{RQ2: Scalability}
\label{sec:eval-rq2}

\F~\ref{fig:scaling} plots analysis time and constraint reduction as each
architecture scales from $1\times$ to $32\times$.

\parh{Analysis time.}
{\emergencystretch=1em
\S\ref{sec:approach} establishes that the abstract interpretation and provenance
construction phases are near-linear in constraint count~$n$, while the commit
loop is $O(n \cdot \bar D)$ with worst case $O(n^2)$. Empirically, fitting $T
\propto n^\alpha$ between the $1\times$ and $32\times$ endpoints yields
exponents $\alpha \in [0.90, 1.70]$ on \texttt{ezkl} and $\alpha \in [1.30, 2.50]$
on \texttt{zkml}. The higher exponents on \texttt{zkml} are consistent with
deeper provenance chains in the dense fully-connected layers (increasing the
per-query cost~$\bar D$) and larger arity per dot-product constraint (increasing
the $d$). In absolute terms, the smallest $1\times$ circuit (zkml/FC-MLP, $476$
constraints) completes in $0.10$\,s, and the largest $32\times$ circuit
(ezkl/MobileNet, $38.40$M constraints) completes in $2.90$\,h. This indicates that
our tool maintains efficiency as circuits scale, and since debloating is a
one-time, offline transformation, even multi-hour analyses amortize over
arbitrarily many subsequent proof generations.\looseness=-1\par}

\parh{Reduction stability.}
On \texttt{zkml}, reduction percentages are remarkably stable across scale:
RNN, GRU, and LSTM vary by less than $0.30$ percentage points from $1\times$ to
$32\times$ ($25.20$\%, $25.00$\%, and $27.30$\%, respectively). FC-MLP and
Transformer show similarly narrow ranges ($27.30$--$27.60$\% and
$19.80$--$21.40$\%). LeNet-CNN increases slightly from $29.20$\% to $31.70$\%.
On \texttt{ezkl}, MobileNet ($21.90$--$23.10$\%) and Transformer
($28.00$--$36.40$\%) are similarly stable or improving; FC-MLP and the recurrent
families show a modest decline ($8$--$10$ percentage points) at high multipliers.

The difference traces to how each pipeline scales. \texttt{zkml} composes
self-contained gadgets, so scaling the principal dimension replicates identical
gadget instances and their internal redundancy in each step. On \texttt{ezkl},
the absolute number of trimmed constraints also grows with scale, but the total
constraint count grows faster for some families (FC-MLP, RNN, GRU, LSTM),
diluting the percentage. For MobileNet and Transformer, total and trimmed
constraints scale at similar rates, keeping the ratio stable or improving.\looseness=-1

\begin{rqanswer}
\textbf{Answer to RQ2.}
Analysis time scales efficiently with constraint count, and reduction
percentages stay stable across scale on both pipelines.
\end{rqanswer}

\subsection{RQ3: Comparison}
\label{sec:eval-rq3}

\begin{table}[t]
\centering
\caption{RQ3: Comparison with existing constraint-reduction tools on five
architectures.}
\label{tab:rq3}
\small
\setlength{\tabcolsep}{3pt}
\begin{tabular}{lr|rrrrr|rr}
\toprule
Circuit & $|C|$ & O2 & Distil. & CirC & \reved{Clap} & \reved{korrekt} & No Graph & Ours \\
\midrule
MLP & 5,463 & 0\% & 0\% & 0\% & 0\% & 0\% & timeout & 22.9\% \\
LeNet & 60.2K & 0\% & 0\% & 0.1\% & 0\% & 0\% & timeout & 29.6\% \\
RNN & 19.7K & 0.2\% & 0.2\% & 0.4\% & 0\% & 0\% & timeout & 15.9\% \\
GRU & 10.8K & 0\% & 0\% & 0.2\% & 0\% & 0\% & timeout & 42.8\% \\
LSTM & 10.8K & 0\% & 0\% & 0.3\% & 0\% & 0\% & 51.6\% & 51.6\% \\
\bottomrule
\end{tabular}
\end{table}

Table~\ref{tab:rq3} reports the constraint reduction achieved by each tool on
the five RQ3 architectures. $|C|$ is the R1CS constraint count (baseline).
``O2'' is Circom~O2; ``Distilling'' is Albert et al.'s approach; ``CirC'' is
CirC's IR optimizations; \reved{``Clap'' and ``korrekt'' run on the
\texttt{zkml} compilation of the same architectures}; and ``No Graph'' is our
debloater with the provenance graph disabled, falling back to per-candidate
re-verification (\S\ref{sec:eval-setup}). The final column is our full
debloater.\looseness=-1

\parh{Circom~O2 and Distilling.}
Circom~O2 and Distilling achieve no reduction on four of the five circuits. The
sole exception is PTB-RNN, where both remove $0.2$\% of constraints, which is a
negligible improvement. This is expected: both tools target algebraic
simplifications (constant folding, linear-combination merging) that do not
reason about cross-constraint entailment. \looseness=-1

\parh{CirC.}
We render each circuit in ZoKrates and compile through CirC with and without
IR optimizations, including constant folding, variable scalarization, and
oblivious array elimination. CirC achieves at most $0.4$\% reduction: these ML
circuits lack the program features (oblivious arrays, hashing, RAM) targeted by
its passes.\looseness=-1

\parh{\reved{Clap and Halo2-Analyzer (korrekt).}}
\reved{Both tools remove nothing on any circuit (Table~\ref{tab:rq3}). The zeros are
structural, not incidental. Korrekt flags gates that are \emph{statically
dead}: polynomials that abstractly evaluate to zero on every active row via
zero constants, zero scale factors, or known-zero fixed cells. We verified
that every gate expression in these circuits abstract-evaluates to
non-constant, because \texttt{zkml} activates each gate template exactly on
the rows it uses and embeds only non-zero weights. Clap's range-check
elimination drops a check only when a syntactically identical check on the
same signal appears elsewhere; after equality merging, every lookup in these
circuits enforces a distinct expression, leaving nothing to
drop. Both tools ask whether a check is ever active or
duplicated; neither asks whether an active, unique check is entailed, which
is where the redundancy resides (22.3--48.7\% on the same circuits,
Table~\ref{tab:rq1}).}\looseness=-1

\parh{\reved{Exact SMT (QF\_NIA).}} \reved{At the scale of Table~\ref{tab:rq3},
the exact method is out of reach. We therefore run it on scaled-down instances
of the five RQ3 architectures (8--680 neurons; 116--7.2K constraints). Within
the 24-hour budget, only the
smallest instance, MLP (116 constraints), terminates, and it takes 11.4 hours to
do so. Its removal set exactly matches the one our tool computes in 0.06\,s. The
other four exceed the budget. Exact entailment is thus the right specification
but not a scalable algorithm: our provenance-guided approximation reaches the
same answer five orders of magnitude faster.}\looseness=-1

\parh{No-Graph ablation.}
On LSTM, No-Graph reaches the same $51.6$\% reduction (in $42$~min) as our full
debloater; on the other four circuits, it exhausts the 4-hour budget, each
candidate triggering a fixpoint recomputation over the entire constraint set.
LSTM is the outlier for a structural reason: its number of constraints
entering analysis (total minus equality constraints) is the smallest in the
corpus ($966$ versus $1{,}140$--$13{,}014$ for the others), so each
per-candidate fixpoint is correspondingly cheap. Our provenance graph replaces
every fixpoint with an on-graph proof search over the candidate's premises,
finishing all five circuits in $4$\,s (LSTM) to $11$\,min (LeNet). At
larger scale in RQ1 (tens of millions of constraints), the same $O(n^2)$ cost
may take days; the graph is what makes that scale tractable.\looseness=-1

\parh{Our tool.}
Our tool reduces constraints by $15.90$\% (RNN) to $51.60$\% (LSTM), the only
substantial reductions within budget across all five circuits. GRU and LSTM
compose multiple checked gates over shared hidden-state signals per timestep,
creating denser entailment opportunities than simpler recurrent and
feed-forward circuits.\looseness=-1

\begin{rqanswer}
\textbf{Answer to RQ3.}
\reved{Syntactic tools achieve at most $0.40$\% reduction (Circom~O2, Distilling,
CirC) or none at all (Clap, korrekt); exact SMT terminates only on a
116-constraint instance, confirming our result after 11.4 hours.} Our tool
alone removes redundant checks at substantial rates, achieving
$15.90$--$51.60$\% reduction.
\end{rqanswer}

\section{Discussion}
\label{sec:discussion}

\parh{Incompleteness.}
Exact maximum debloating is intractable at our target scale
(Challenge~C, \S\ref{sec:obstacles}). Our sound but incomplete analysis uses
interval, known-bits, and constant domains, template-based entailment rules,
and greedy removal. These polynomial-time approximations may miss redundant
checks, but every removal remains justified by provenance-guided proof search
over the surviving constraints. Missed removals affect optimization, not
soundness; our evaluation demonstrates substantial reductions on circuits with
tens of millions of constraints.\looseness=-1

\parh{Scope.}
Our tool targets cross-constraint entailment among soundness checks, which is
the dominant redundancy class in ML inference circuits. Algebraic
simplifications (constant folding, linear-combination merging, common
subexpression elimination) lie outside this scope; tools in
\S\ref{sec:prior-opt} already cover them. The two classes touch disjoint
constraints, so the passes can compose.\looseness=-1

\parh{Generality on Non-ML Domains.}
\reved{In principle, our method is general: it consumes only R1CS and PLONKish
constraint systems over $\Fp$, the generic formats that ZK frameworks emit
regardless of the source domain. Only the \emph{density} of redundancy is
ML-specific, since gadget-templated compilation replicates the range, sign, and
decomposition checks that drive our gains. The same checks recur beyond ML,
although our evidence covers only ML circuits:
signature-verification~\cite{circomecdsa} and bignum arithmetic
circuits~\cite{xjsnark} range-check limb decompositions in the same way, so they
are natural next targets. The achievable reduction there is likely smaller,
since field arithmetic and Merkle-path hashing, rather than such checks, account
for most of their constraints. Porting stays within the framework: a new domain
needs abstract domains and transfer functions matched to its invariants, which
we leave to future work.}\looseness=-1

\parh{Scalability on Complex ML Models.}
\reved{Modern ZK ML workloads increasingly target transformers and large
language models. These compile from the same operators as our benchmarks and
therefore emit the same checks, so we expect the method to extend to them.
Consistent with this, BERT-tiny yields $24.90$--$37.40$\% reduction
(\S\ref{sec:eval-rq1}), stable as the circuit scales $32\times$
(\S\ref{sec:eval-rq2}). We presume the trend continues beyond this scale, though
end-to-end evaluation is currently precluded by prover memory. It is a
bottleneck for ZK ML at large rather than for our analysis, with GPT-2 alone
requiring roughly 1\,TB of RAM~\cite{zkml}.}\looseness=-1

\section{Related Work}
\label{sec:related}


\parh{ZK machine learning and acceleration.} A broad literature authors circuits
for specific model families with custom gadgets and protocol choices:
zkCNN~\cite{zkcnn}, vCNN~\cite{vcnn}, and pvCNN~\cite{pvcnn} for CNNs;
zkLLM~\cite{sun2024zkllmzeroknowledgeproofs} for language models; and works on
decision-tree predictions~\cite{zkdecisiontree}, proofs of
training~\cite{zkpot}, and R1CS gadgets for NN primitives~\cite{Feng2021ZENEZ}.
A complementary line reduces proving cost without post-processing the
constraints: folding schemes~\cite{nova,hypernova} amortize repeated structure;
lookup arguments~\cite{plonkup,logup,lasso,jolt} offload nonlinearities to
tables; GKR-based protocols~\cite{libra, hyrax} prove layered computations via
sumcheck; and hardware work~\cite{gzkp,zkpog} lowers prover cost directly. These
choices are made at authoring time or per-operation; our debloater is a
post-pass that removes constraints already redundant by composition, and is
therefore orthogonal.\looseness=-1

\parh{DNN verification.} Neural-network verifiers use abstract interpretation
and bound propagation to prove input-output properties:
AI$^2$~\cite{gehr2018ai2} introduced abstract interpretation for certification;
Reluplex~\cite{katz2017reluplex} and Marabou~\cite{katz2019marabou} combine
Simplex with ReLU case splitting; ERAN~\cite{singh2019abstract, singh2018fast}
uses zonotope and polyhedral abstractions; and
$\alpha,\beta$-CROWN~\cite{zhang2018efficient, xu2021fast, wang2021beta}
combines linear bounds with branch-and-bound; and
DeepSRGR~\cite{yang2021spurious} tightens neuron bounds via LP. We share the use
of abstract domains, but our object is the compiled ZK constraint system
$\mathcal{C}$ over $\Fp$, with lookups, copy constraints, and decompositions
that have no direct NN analogue, and our question is constraint entailment:
whether $c$ follows from $\mathcal{C}\setminus\{c\}$, not whether a network
satisfies an input-output property.\looseness=-1

\parh{Program analysis on ZK circuits.}
Static analysis has been applied to ZK circuits primarily for soundness.
QED$^2$~\cite{picus} combines unique-constraint propagation with SMT to flag
\emph{under-constrained} circuits, finding genuine vulnerabilities in production
Circom; Coda~\cite{Liu2023CertifyingZC} uses a refinement-typed language to
uncover bugs in widely-used libraries; and
AC\textsuperscript{4}~\cite{ac4} and Ecne~\cite{ecne} detect
under-constraint patterns algebraically or by uniqueness propagation. These
tools ask whether the witness set of $\mathcal{C}$ is too large; we ask whether
individual constraints contribute to keeping it small. The analyses are
therefore complementary.\looseness=-1

\parh{Software debloating.}~\emph{Debloating} is borrowed from work that removes
unused code or runtime checks. Trimmer~\cite{trimmer}, RAZOR~\cite{razor}, and
Chisel~\cite{heo2018chisel} trim against configurations, representative traces,
or behavioral specifications, and BB-DSE~\cite{bardin2017backward} strips
infeasible branches in obfuscated binaries; SanRazor~\cite{jiang2021sanrazor}
and ASAP~\cite{asap} target redundant sanitizer checks. Closer in shape is
LClean~\cite{marcozzi2018time}, which prunes redundant test-coverage labels by
per-label SMT entailment (sound, incomplete). All these tools operate over
control-flow IRs. Our setting has no control flow and requires exact witness-set
preservation: each removal is a one-sided-sound entailment discharged via a
precomputed AND/OR provenance graph rather than per-candidate SMT, scaling to
$10^7$ constraints.\looseness=-1


\section{Conclusion}
\label{sec:conclusion}

We presented a sound debloater combining whole-circuit abstract interpretation
with provenance-guided removal. Across \texttt{ezkl} and \texttt{zkml} circuits
from four architecture families, it preserves the witness set while removing
25.5\% of constraints and reducing proving time by 33.4\% on average.
\looseness=-1

\begin{acks}
The HKUST authors were supported in part by a grant from the Research Grants
Council of the Hong Kong Special Administrative Region, China HKUST
(No.~R6005-25) and a RGC GRF grant under the contract 16214723.
\end{acks}


\appendix

\section{Open Science}
\label{sec:open-science}

{\emergencystretch=1em
Our complete Rust implementation and evaluation artifact are available at
\url{https://github.com/xuezhantong/zk-debloater}.\par}

\parh{Artifacts provided.}
The bundle contains the tool; model generators and halo2-to-JSON extractor; R1CS
corpus; prover/verifier sources; and drivers and results for every experiment in
\S\ref{sec:eval}.

\section{Ethical Considerations}
\label{sec:ethics}

This work involves no human subjects, personal data, or vulnerability
disclosure. Its static analysis removes only provably redundant constraints;
we foresee no negative ethical implications.

\section{Generative AI Usage}
\label{sec:generative}

We used OpenAI GPT and Anthropic Claude for proofreading, rephrasing,
formatting, debugging, and implementation suggestions. The authors reviewed and
verified all generated content and code for accuracy, correctness, and
originality.

\section{Proof Sketches}
\label{sec:appendix-proofs}

We sketch the complexity theorem of \S\ref{sec:provenance} and the soundness
theorems of \S\ref{sec:debloating}.

\begin{proof}[Proof sketch of Theorem~\ref{thm:prov-complexity}]%
\emergencystretch=1em
Each $(c,v)$ pair uses one $O(d)$ transfer replay and $O(d)$ premise trials of
$O(d)$ each, hence $O(d^2)$ per pair and $O(nd^3)$ overall. The graph bounds
follow from $|\mathcal{K}|=3$ and at most $d-1$ premises per constraint node.
\end{proof}

\begin{proof}[Proof sketch of Theorem~\ref{thm:prov-soundness}]
By induction on the proof tree of $f$ under $M_R$. Each internal node
$\nu=(c,(v,k),P)\notin M_R$ has $c\in\mathcal{C}\setminus R$, so $w$ satisfies
$c$. For an axiom node ($P=\emptyset$),
\Alg~\ref{alg:build-prov-graph} materializes precisely when
$\mathrm{transfer}_c(\top)(v)\Vdash(v,k)$. Since $\top$ over-approximates $w$,
local soundness makes $w(v)$ satisfy $k$. Inductively, smaller subtrees prove
that $w$ satisfies every $p\in P$. Thus $\mathit{state}(P)$ over-approximates
$w$; the replay that emitted $(v,k)$ (\Alg~\ref{alg:min-premises}) and local
soundness again imply that $w(v)$ satisfies $k$.
\end{proof}

\begin{proof}[Proof sketch of Theorem~\ref{thm:soundness}]
Let $c_1,\ldots,c_k$ witness the ordering. The forward inclusion is immediate.
For the reverse, take $w\models\mathcal{C}\setminus\mathcal{R}_i$.
By \Def~\ref{def:grounded}, every fact in $\Obl(c_i)$ is provable under
$M_{\mathcal{R}_i}$ and therefore holds by Theorem~\ref{thm:prov-soundness}.
Entailment soundness gives $w\models c_i$, hence
$w\models\mathcal{C}\setminus\mathcal{R}_{i-1}$. Iterating to
$\mathcal{R}_0=\emptyset$ yields $w\models\mathcal{C}$.
\end{proof}

\section{Illustrative Walkthrough}

\label{sec:walkthrough}

\begin{figure}[t]
\centering

\resizebox{\columnwidth}{!}{%
\begin{tikzpicture}[
  font=\scriptsize,
  >={Latex[length=1.4mm, width=1mm]},
  orfact/.style={draw, rounded corners, fill=blue!5,
                 inner xsep=2pt, inner ysep=1.5pt,
                 font=\scriptsize, minimum height=4mm},
  axand/.style={draw, circle, fill=blue!15,
                inner sep=0pt, minimum size=5mm,
                font=\tiny, text=blue!55!black},
  jnand/.style={draw, circle, fill=orange!25,
                inner sep=0pt, minimum size=5mm,
                font=\tiny, text=orange!75!black},
  axedge/.style={->, thin, blue!55!black},
  conedge/.style={->, thin, orange!75!black},
  preedge/.style={->, densely dashed, thin, black!55},
]

\def\dx{1.6}
\def\axY{1.25}
\def\jnY{-1.25}

\node[orfact] (zlo) at (0*\dx, 0) {$z\!\geq\!0$};
\node[orfact] (zhi) at (1*\dx, 0) {$z\!\leq\!5$};
\node[orfact] (xlo) at (2*\dx, 0) {$x\!\geq\!0$};
\node[orfact] (yhi) at (3*\dx, 0) {$y\!\leq\!10$};
\node[orfact] (xhi) at (4*\dx, 0) {$x\!\leq\!10$};
\node[orfact] (ylo) at (5*\dx, 0) {$y\!\geq\!0$};

\node[axand] (axRelo) at (0*\dx-0.38, \axY) {$c_{\mathsf{relu}}$};
\node[axand] (axRnlo) at (0*\dx+0.38, \axY) {$c_{\mathsf{rng}}$};
\node[axand] (axRehi) at (1*\dx-0.38, \axY) {$c_{\mathsf{relu}}$};
\node[axand] (axRnhi) at (1*\dx+0.38, \axY) {$c_{\mathsf{rng}}$};
\node[axand] (axXlo)  at (2*\dx,      \axY) {$c_{x}$};
\node[axand] (axYhi)  at (3*\dx,      \axY) {$c_{y}$};
\node[axand] (axXhi)  at (4*\dx,      \axY) {$c_{x}$};
\node[axand] (axYlo)  at (5*\dx,      \axY) {$c_{y}$};

\node[jnand] (jnxlo) at (2*\dx, \jnY) {$c_{\mathsf{sum}}$};
\node[jnand] (jnyhi) at (3*\dx, \jnY) {$c_{\mathsf{sum}}$};
\node[jnand] (jnxhi) at (4*\dx, \jnY) {$c_{\mathsf{sum}}$};
\node[jnand] (jnylo) at (5*\dx, \jnY) {$c_{\mathsf{sum}}$};

\draw[axedge] (axRelo) -- (zlo);
\draw[axedge] (axRnlo) -- (zlo);
\draw[axedge] (axRehi) -- (zhi);
\draw[axedge] (axRnhi) -- (zhi);
\draw[axedge] (axXlo)  -- (xlo);
\draw[axedge] (axYhi)  -- (yhi);
\draw[axedge] (axXhi)  -- (xhi);
\draw[axedge] (axYlo)  -- (ylo);

\draw[conedge] (jnxlo) -- (xlo);
\draw[conedge] (jnyhi) -- (yhi);
\draw[conedge] (jnxhi) -- (xhi);
\draw[conedge] (jnylo) -- (ylo);

\draw[preedge] (yhi.south) -- (jnxlo.north);
\draw[preedge] (xlo.south) -- (jnyhi.north);
\draw[preedge] (ylo.south) -- (jnxhi.north);
\draw[preedge] (xhi.south) -- (jnylo.north);

\end{tikzpicture}%
}

\par\smallskip
{\scriptsize
  \tikz[baseline=-0.5ex]{\node[draw, rounded corners=1pt,
       fill=blue!5, minimum width=4mm, minimum height=2.2mm, inner sep=0]{};}~OR-fact\enspace
  \tikz[baseline=-0.5ex]{\node[draw, circle, fill=blue!15,
       minimum size=2.4mm, inner sep=0]{};}~axiom AND\enspace
  \tikz[baseline=-0.5ex]{\node[draw, circle, fill=orange!25,
       minimum size=2.4mm, inner sep=0]{};}~joint AND\enspace
  \tikz[baseline=-0.5ex]{\draw[->, blue!55!black, semithick]
       (0,0) -- (0.45,0);}~axiom concl.\enspace
  \tikz[baseline=-0.5ex]{\draw[->, orange!75!black, semithick]
       (0,0) -- (0.45,0);}~joint concl.\enspace
  \tikz[baseline=-0.5ex]{\draw[->, densely dashed, black!55, semithick]
       (0,0) -- (0.45,0);}~premise
}

\par\smallskip
{\small \textit{(a) Phase~2: fixpoint provenance graph.}}

\vspace{3mm}

\resizebox{\columnwidth}{!}{%
\begin{tikzpicture}[
  font=\scriptsize,
  >={Latex[length=1.4mm, width=1mm]},
  orfact/.style={draw, rounded corners, fill=blue!5,
                 inner xsep=2pt, inner ysep=1.5pt,
                 font=\scriptsize, minimum height=4mm},
  goalfact/.style={draw, rounded corners, fill=green!22, thick,
                   inner xsep=2pt, inner ysep=1.5pt,
                   font=\scriptsize, minimum height=4mm},
  axand/.style={draw, circle, fill=blue!15,
                inner sep=0pt, minimum size=5mm,
                font=\tiny, text=blue!55!black},
  jnand/.style={draw, circle, fill=orange!25,
                inner sep=0pt, minimum size=5mm,
                font=\tiny, text=orange!75!black},
  subtree/.style={draw=green!55!black, fill=green!12,
                  thick, line cap=round},
  axedge/.style={->, thin, blue!55!black},
  conedge/.style={->, thin, orange!75!black},
  preedge/.style={->, densely dashed, thin, black!55},
  liveedge/.style={->, very thick, green!45!black},
  livepre/.style={->, densely dashed, very thick, green!45!black},
  mask/.style={red!75!black, ultra thick, line cap=round},
]

\def\dx{1.6}
\def\axY{1.25}
\def\jnY{-1.25}

\node[orfact] (zlo) at (0*\dx, 0) {$z\!\geq\!0$};
\node[orfact] (zhi) at (1*\dx, 0) {$z\!\leq\!5$};
\node[orfact] (xlo) at (2*\dx, 0) {$x\!\geq\!0$};
\node[orfact] (yhi) at (3*\dx, 0) {$y\!\leq\!10$};
\node[goalfact] (xhi) at (4*\dx, 0) {$x\!\leq\!10$};
\node[orfact] (ylo) at (5*\dx, 0) {$y\!\geq\!0$};

\node[axand] (axRelo) at (0*\dx-0.38, \axY) {$c_{\mathsf{relu}}$};
\node[axand] (axRnlo) at (0*\dx+0.38, \axY) {$c_{\mathsf{rng}}$};
\node[axand] (axRehi) at (1*\dx-0.38, \axY) {$c_{\mathsf{relu}}$};
\node[axand] (axRnhi) at (1*\dx+0.38, \axY) {$c_{\mathsf{rng}}$};
\node[axand] (axXlo)  at (2*\dx,      \axY) {$c_{x}$};
\node[axand] (axYhi)  at (3*\dx,      \axY) {$c_{y}$};
\node[axand] (axXhi)  at (4*\dx,      \axY) {$c_{x}$};
\node[axand] (axYlo)  at (5*\dx,      \axY) {$c_{y}$};

\node[jnand] (jnxlo) at (2*\dx, \jnY) {$c_{\mathsf{sum}}$};
\node[jnand] (jnyhi) at (3*\dx, \jnY) {$c_{\mathsf{sum}}$};
\node[jnand] (jnxhi) at (4*\dx, \jnY) {$c_{\mathsf{sum}}$};
\node[jnand] (jnylo) at (5*\dx, \jnY) {$c_{\mathsf{sum}}$};

\begin{scope}[on background layer]
\draw[subtree]
  plot [smooth cycle, tension=0.6] coordinates {
    ($(axYlo.north east) + (+4mm, +4mm)$)
    ($(ylo.south east)   + (+4mm, -4mm)$)
    ($(ylo.south west)   + (-1mm, -5mm)$)
    ($(jnxhi.south east) + (+5mm, -3mm)$)
    ($(jnxhi.south west) + (-4mm, -4mm)$)
    ($(xhi.north west)   + (-4mm, +4mm)$)
    ($(xhi.north east)   + (+1mm, +5mm)$)
  };
\end{scope}

\draw[axedge] (axRelo) -- (zlo);
\draw[axedge] (axRnlo) -- (zlo);
\draw[axedge] (axRehi) -- (zhi);
\draw[axedge] (axRnhi) -- (zhi);
\draw[axedge] (axXlo)  -- (xlo);
\draw[axedge] (axYhi)  -- (yhi);
\draw[axedge] (axXhi)  -- (xhi);
\draw[liveedge] (axYlo)  -- (ylo);

\draw[conedge] (jnxlo) -- (xlo);
\draw[conedge] (jnyhi) -- (yhi);
\draw[liveedge] (jnxhi) -- (xhi);
\draw[conedge] (jnylo) -- (ylo);

\draw[preedge] (yhi.south) -- (jnxlo.north);
\draw[preedge] (xlo.south) -- (jnyhi.north);
\draw[livepre] (ylo.south) -- (jnxhi.north);
\draw[preedge] (xhi.south) -- (jnylo.north);

\foreach \n in {axRelo, axRehi, axXlo, axXhi}{
  \draw[mask] (\n.north west) -- (\n.south east);
  \draw[mask] (\n.north east) -- (\n.south west);
}

\end{tikzpicture}%
}

\par\smallskip
{\scriptsize
  \tikz[baseline=-0.5ex]{\node[draw, circle, fill=blue!15,
       minimum size=2.4mm, inner sep=0]{};\draw[red!75!black, ultra thick, line cap=round]
       (-1.2mm,-1.2mm) -- (1.2mm,1.2mm);\draw[red!75!black, ultra thick, line cap=round]
       (-1.2mm,1.2mm) -- (1.2mm,-1.2mm);}~masked\enspace
  \tikz[baseline=-0.5ex]{\draw[draw=green!55!black, fill=green!12, thick]
       (0,0) ellipse (2.5mm and 1.4mm);}~subtree\enspace
  \tikz[baseline=-0.5ex]{\draw[->, green!45!black, very thick]
       (0,0) -- (0.45,0);}~live edge
}

\par\smallskip
{\small \textit{(b) Phase~3: trial of $c_x$ on top of $c_{\mathsf{relu}}$.}}

\caption{Provenance graph for the walkthrough (\S\ref{sec:walkthrough}).}
\Description{Two stacked panels showing the same six-fact provenance graph for
the walkthrough slice. Panel (a) is the fixpoint graph: six OR-facts in the
middle row, eight axiom AND-nodes above, four joint $c_{\mathsf{sum}}$ AND-nodes
below, with dashed premise edges forming two crossing-X cycles among the $x,y$
facts. Panel (b) overlays the Phase~3 trial: the four axiom AND-nodes for
$c_{\mathsf{relu}}$ and $c_x$ are struck through with red crosses, and a
green-tinted subtree encloses the surviving proof of $x\!\leq\!10$ via
$c_{\mathsf{sum}}$ and $c_y$.}
\Description{Two stacked panels showing the same six-fact provenance graph
  for the walkthrough slice. Panel (a) is the fixpoint graph: six OR-facts in
  the middle row, eight axiom AND-nodes above, and four joint sum-constraint
  AND-nodes below, with dashed premise edges forming two crossing cycles among
  the x and y facts. Panel (b) overlays the third-phase trial: the four axiom
  AND-nodes for the ReLU and x constraints are struck through, and a green
  subtree encloses the surviving proof of the upper bound on x through the sum
  and y constraints.}
\label{fig:walkthrough}
\end{figure}
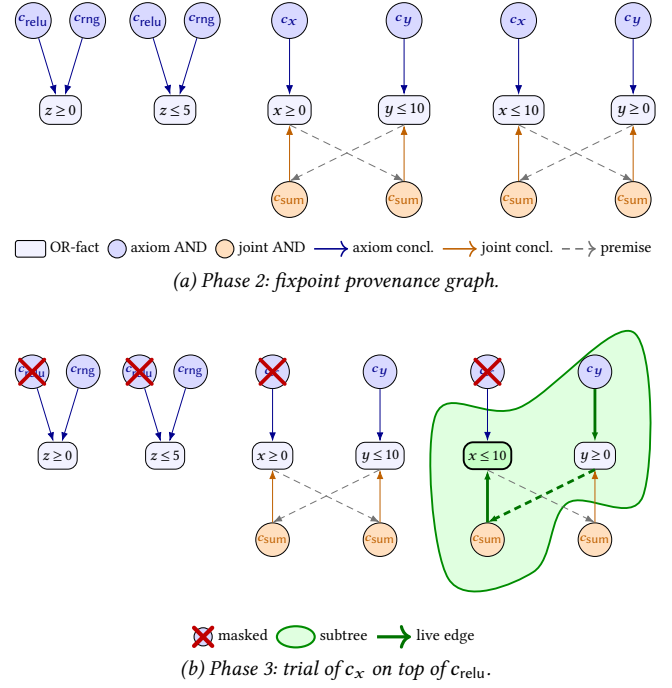

We trace interval facts through five constraints distilled from
\S\ref{sec:motivating-example}: $c_{\mathsf{relu}}:z\in[0,10]$,
$c_{\mathsf{rng}}:z\in[0,5]$, $c_{\mathsf{sum}}:x+y=10$, and
$c_x:x\in[0,10]$, $c_y:y\in[0,10]$.
The $\{c_{\mathsf{relu}}, c_{\mathsf{rng}}\}$ and
$\{c_{\mathsf{sum}}, c_x, c_y\}$ subsystems share no variables.

\parh{Phase 1: Abstract interpretation.}
Intersecting $c_{\mathsf{relu}}$ and $c_{\mathsf{rng}}$ tightens $z$ to
$[0, 5]$; propagating a partner bound through $c_{\mathsf{sum}}$ yields no
tightening of $x$ or $y$ beyond their lookups. The fixpoint is
$\sigma^\star(z) = [0, 5]$ and $\sigma^\star(x) = \sigma^\star(y) = [0, 10]$.

\parh{Phase 2: Provenance graph.}
\F~\ref{fig:walkthrough}(a) shows the result. \emph{Axiom} nodes (blue)
record lookups deriving their fact unconditionally: $c_{\mathsf{rng}}$ alone
proves $z \le 5$; $c_x$ alone proves $0 \le x \le 10$; similarly for $c_y$.
\emph{Joint} nodes (orange) record $c_{\mathsf{sum}}$ deriving a fact from a
partner premise: $x + y = 10$ together with $y \ge 0$ yields $x \le 10$,
recorded as a premise edge from the $y \ge 0$ fact node into the joint node.
The four joint nodes form two 2-cycles among the $x, y$ facts.

\parh{Phase 3: Debloating.}
The filter admits the four lookups but rejects $c_{\mathsf{sum}}$, whose gate
rule requires constants. The loop removes $c_{\mathsf{relu}}$ because
$c_{\mathsf{rng}}$ proves $z\le10$, then retains $c_{\mathsf{rng}}$ because no
survivor targets $z$. It removes $c_x$ because $c_{\mathsf{sum}}$ and $c_y$
prove its obligations (\F~\ref{fig:walkthrough}(b)). Finally, proving $c_y$
would require $c_x$ and then $c_y$ itself; cycle detection
(\Alg~\ref{alg:proof-search}) therefore retains it.
The final removal set is $\mathcal{R} = \{c_{\mathsf{relu}}, c_x\}$.

\begingroup
\emergencystretch=1em
\section{Case Studies of Removed Constraints}
\label{sec:appendix-case-studies}

\S\ref{sec:motivation} sketches a small illustrative example of how such
redundancies arise in a ZK ML pipeline fragment. This appendix walks through
three real cases drawn from the benchmarks of Table~\ref{tab:rq1}, each
isolating a distinct removal mechanism; together they cover the patterns we
observe across both pipelines.

\subsection{Range subsumption (zkml/MNIST-MLP)}
\label{app:case-subsume}

Of the $791$ removed constraints, $788$ are \emph{lookups} and $3$ are gates.
Each removed lookup is a \emph{range proof} on a variable whose range is
already pinned to an at-least-as-tight interval by some other constraint in
the circuit. Three patterns recur on this benchmark and account for nearly
all $788$ lookup trims (Fig.~\ref{fig:range-subsume}).

\begin{figure}[h]
\centering
\begin{tikzpicture}[
  font=\footnotesize,
  box/.style={draw, rounded corners=1.5pt, inner sep=2.4pt, minimum height=4.4mm},
  redu/.style={draw, dashed, rounded corners=1.5pt, inner sep=2.4pt,
               text=gray!70, minimum height=4.4mm},
  arr/.style={->, >=stealth},
  caplab/.style={font=\scriptsize},
  lbl/.style={font=\itshape}
]
\node[box]                       (a1) {ReLU output lookup};
\node[redu, right=22mm of a1]    (a2) {range lookup};
\draw[arr] (a1) -- node[above, caplab]{$\mathit{act}\in[0, 128]$}
                   node[below, caplab, gray!70]{$\subset[0, 256]$} (a2);
\node[lbl, anchor=east] at ($(a1.west) + (-2mm, 0)$) {(a)};

\node[box, below=8mm of a1]      (b1) {ReLU output lookup};
\node[redu, right=22mm of b1]    (b2) {Packer lookup};
\draw[arr] (b1) -- node[above, caplab]{$x\in[0, 128]$}
                   node[below, caplab, gray!70]{$x{+}128\in[0, 256]$} (b2);
\node[lbl, anchor=east] at ($(b1.west) + (-2mm, 0)$) {(b)};

\node[box, below=11mm of b1]     (c1) {gate $b{=}16$};
\node[box, below=2.5mm of c1]    (c2) {$2b{-}r$ lookup};
\node[redu] (c3) at ($(c1.east)!0.5!(c2.east) + (3.4cm, 0)$) {$r$ lookup};
\draw[arr] (c1.east) -- ([yshift=1mm]c3.west);
\draw[arr] (c2.east) -- ([yshift=-1mm]c3.west);
\node[caplab,fill=white,inner sep=1pt] at ($(c1.east)!0.5!(c2.east) + (1.7cm, 0)$) {$\Rightarrow r\in[0, 32]$};
\node[lbl, anchor=east] at ($(c1.west)!0.5!(c2.west) + (-2mm, 0)$) {(c)};
\end{tikzpicture}
\caption{Three range-subsumption patterns. {Dashed boxes are the redundant
  lookups.}}
\Description{Three patterns: a ReLU bound subsumes a wider range lookup;
a ReLU bound subsumes a shifted Packer lookup; and a fixed divisor together
with a two-input lookup entails a remainder lookup. Dashed boxes mark removals.}
\label{fig:range-subsume}
\end{figure}
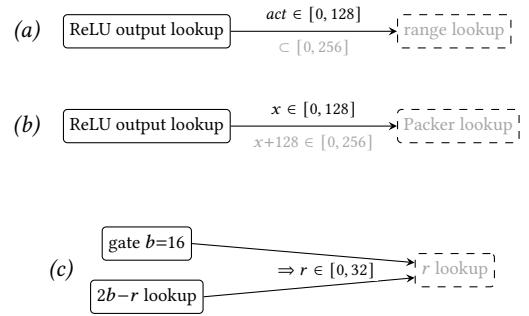

Pattern~(a). The ReLU gadget contains a lookup table that, by
construction, pins its output variable $\mathit{act}$ to $[0, 128]$.
Further downstream, the compiler emits a separate, generic range check
on the same variable against the wider interval $[0, 256]$. The first
lookup already implies the second, so the second is removable.

Pattern~(b). The ReLU gadget's output lookup pins the post-activation
variable $x$ to $[0, 128]$. zkml's polynomial Packer (the inter-layer
commitment gadget that runs after each ReLU) then commits to the same
$x$ by range-checking the shifted expression $x{+}128$ against
$[0, 256]$. The ReLU's existing bound already implies
$x{+}128 \in [128, 256] \subset [0, 256]$, so the Packer's check is
redundant.

Pattern~(c). Within a single VarDiv invocation, three constraints
jointly bound the remainder $r$: a gate fixing the divisor $b{=}16$, a
lookup on the multi-input expression $2b{-}r$, and a standalone lookup
on $r$. The first two together force $r \in [0, 32]$, so the standalone
$r$-lookup adds nothing. When the divisor's constant value is
independently provable (e.g., from a shared fixed cell that multiple
VarDiv invocations reach through copy chains), the divisor-fixing gate
itself becomes entailed and is dropped along with the lookup, removing
two constraints per such VarDiv rather than one.


\subsection{Activation-bound cascade (ezkl/BERT-tiny)}
\label{app:case-axiom}

On BERT-tiny ($2.87$\,M constraints after equality merging), the debloater
removes $715{,}517$ constraints, with gate trims ($404{,}924$) outnumbering
lookup trims ($310{,}593$). Because ezkl keeps model weights in private
advice columns, their values remain unknown to the analysis, and all
$309{,}696$ dot-product gates that compute the network's weighted sums
survive verbatim. The trims trace to one of two cascades, both seeded by an
upstream range bound on a variable~$x$ (Fig.~\ref{fig:bert-cascade}).

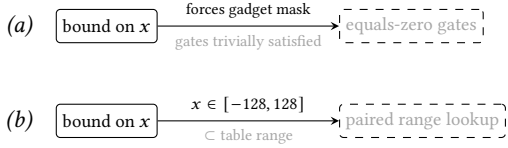
\begin{figure}[h]
\centering
\begin{tikzpicture}[
  font=\footnotesize,
  box/.style={draw, rounded corners=1.5pt, inner sep=2.4pt, minimum height=4.4mm},
  redu/.style={draw, dashed, rounded corners=1.5pt, inner sep=2.4pt,
               text=gray!70, minimum height=4.4mm},
  arr/.style={->, >=stealth},
  caplab/.style={font=\scriptsize},
  lbl/.style={font=\itshape}
]
\node[box]                       (a1) {bound on $x$};
\node[redu, right=24mm of a1]    (a2) {equals-zero gates};
\draw[arr] (a1) -- node[above, caplab]{forces gadget mask}
                   node[below, caplab, gray!70]{gates trivially satisfied} (a2);
\node[lbl, anchor=east] at ($(a1.west) + (-2mm, 0)$) {(a)};

\node[box, below=8mm of a1]      (b1) {bound on $x$};
\node[redu, right=24mm of b1]    (b2) {paired range lookup};
\draw[arr] (b1) -- node[above, caplab]{$x\in[-128, 128]$}
                   node[below, caplab, gray!70]{$\subset$ table range} (b2);
\node[lbl, anchor=east] at ($(b1.west) + (-2mm, 0)$) {(b)};
\end{tikzpicture}
\caption{Two cascades on ezkl BERT-tiny, both seeded by an upstream
range bound on the decomposed variable $x$.}
\Description{An upstream bound either fixes an equals-zero gadget mask,
trivializing its gates, or subsumes a paired range lookup.}
\label{fig:bert-cascade}
\end{figure}

Pattern~(a) accounts for $303{,}485$ of the gate trims. Each lookup-table
activation in ezkl is wrapped in a small equals-zero gadget that produces
a Boolean mask indicating whether the activation's input is zero, used
downstream by sign-handling logic. Once an upstream range check on $x$
either excludes zero or pins $x$ to zero, the mask collapses to a single
value, and the gadget's two enforcing gates reduce to trivial identities
(e.g., $x \cdot \mathit{mask} = 0$ becomes $0 = 0$ once either side is
pinned to zero). The entailment check then discharges both gates.

Pattern~(b) accounts for all $310{,}593$ lookup trims and the remaining
$101{,}439$ gate trims. ezkl decomposes each variable into sign and digit
components and range-checks the components in \emph{matched pairs} of
lookups. The upstream range bound on $x$ subsumes one lookup of each pair,
while the other survives as the lookup that defines the gadget. When the
subsumed lookup's input is additionally pinned to a single constant, the
lookup is first rewritten as an equality gate and then discharged; this
rewrite path accounts for the $101{,}439$ gate trims.

%

\subsection{Recurrent amplification (zkml/PTB-LSTM)}
\label{app:case-recurrent}

PTB-LSTM yields the largest reduction in our corpus: $15{,}240$ of $31{,}280$
post-merge constraints removed. The mechanisms match those of
\S\ref{app:case-subsume}; the amplifier is the LSTM cell's gadget
multiplicity, replicated across timesteps (Fig.~\ref{fig:lstm-unroll}).

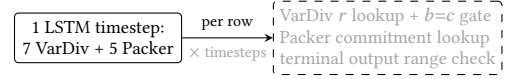
\begin{figure}[h]
\centering
\begin{tikzpicture}[
  font=\footnotesize,
  cell/.style={draw, rounded corners=1.5pt, inner sep=3pt, align=center,
               minimum height=4.4mm},
  redu/.style={draw, dashed, rounded corners=1.5pt, inner sep=2.4pt,
               text=gray!70, align=left, minimum height=4.4mm},
  arr/.style={->, >=stealth},
  caplab/.style={font=\scriptsize}
]
\node[cell] (cell) {1 LSTM timestep:\\\footnotesize 7 VarDiv + 5 Packer};
\node[redu, right=12mm of cell, align=left] (trims)
     {VarDiv $r$ lookup + $b{=}c$ gate\\
      Packer commitment lookup\\
      terminal output range check};
\draw[arr] (cell) -- node[above, caplab]{per row}
                     node[below, caplab, gray!70]{$\times$ timesteps} (trims);
\end{tikzpicture}
\caption{Each LSTM timestep emits the same gadget infrastructure once
per row of the hidden state; the redundancies of \S\ref{app:case-subsume}
fire on every instance.}
\Description{Each LSTM timestep contains seven VarDiv and five Packer
instances. Their removable checks repeat over rows and timesteps.}
\label{fig:lstm-unroll}
\end{figure}

$11{,}200$ trims come from the intra-VarDiv pattern of
\S\ref{app:case-subsume}~(c), $4{,}000$ from the Packer-after-activation
pattern of \S\ref{app:case-subsume}~(b), and the remaining $40$ are
terminal output range checks. What survives across timesteps (the FC
dot-product gates, the gating multiplications, and the activation lookups
themselves) defines the model's actual arithmetic. Recurrent and gated
architectures reach the largest reductions in our corpus because they
unroll proportionally more redundant gadget infrastructure than
feed-forward ones, while the surviving computational core scales similarly.
\par
\endgroup

\section{Empirical Soundness Validation}
\label{sec:appendix-soundness}

Debloating soundness rests on two arguments: the pen-and-paper proof in
\S\ref{sec:approach}, which establishes that any committed removal preserves the
witness set, and the mechanized per-removal proof search built into the
debloater itself, which only commits to removing a constraint $c$ once the
AND/OR provenance graph exhibits an alternative proof for every fact $c$
supports. The study below is an independent \emph{cross-check}: it does not
extend those arguments, but it exercises the end-to-end pipeline (circuit
extraction, entailment heuristics, greedy commit, output rendering) with
adversarial witnesses drawn from boundary-violation and field wrap-around attack
models, to detect implementation bugs in the debloater.

\subsection{Test design}
\label{app:soundness-design}

For each pipeline-model pair we capture two snapshots: the pre-debloat
constraint set $C_{\mathit{orig}}$ and the post-debloat set $C_{\mathit{deb}}$.

\F~\ref{fig:soundness-fuzz} illustrates one iteration. We sample a removed
lookup $\ell \in R = C_{\mathit{orig}} \setminus C_{\mathit{deb}}$, pick an
advice cell $c$ that appears in $\ell$'s input expression, and set $w'[c]$ to an
adversarial value drawn from one of two regimes: an integer-boundary
perturbation just outside $\ell$'s range ($\mathit{lo}{-}1$ or
$\mathit{hi}{+}1$), or a field wrap-around value drawn from $\mathbb{F}_p
\setminus [\mathit{lo}, \mathit{hi}]$ (e.g., $p{-}1$, $p/2$, or $p{-}k$ for
small $k$, exercising the case where a witness satisfies an integer-range
constraint mod~$p$ but escapes the integer range itself). $w'$ agrees with the
original witness $w$ on every other cell. The original rejects $w'$ by
construction, since $\ell$ is violated. Soundness requires that some constraint
in $C_{\mathit{deb}}$ also rejects $w'$; if every $C_{\mathit{deb}}$ constraint
accepts $w'$, the debloater has unsoundly trimmed $\ell$.

We implement the loop in Rust. Each pipeline-model pair runs to a budget of 3
days. Sampling is random over the removed lookup set $R_{\mathit{lookup}}$ but
coverage-prioritized, so every removed lookup is tested at least once.

\parh{Limitations.} This testing pipeline is a sampling cross-check, not a
soundness proof:

(i)~\emph{Single-cell perturbations}: only one advice cell is modified per
iteration, so coordinated multi-cell witnesses that preserve gate equalities
while violating $\ell$ are out of scope.

(ii)~\emph{Sampling, not enumeration}: coverage is at-least-once per removed
lookup, not exhaustive across rows or attack patterns. We accept these gaps
because empirical fuzzing at this scale cannot enumerate $\mathbb{F}_p$;
soundness against worst-case adversaries is underwritten by the proofs above,
not by this study.

\begin{figure}[h]
\centering
\begin{tikzpicture}[
  font=\footnotesize,
  >=Stealth,
  every node/.style={inner sep=3pt},
  box/.style={draw, rounded corners=2pt, align=center,
              minimum width=2.3cm, minimum height=0.9cm}
]
  \node[box, fill=red!7, draw=red!50!black] (w)
       {adversarial witness $w'$\\(cell $c$ of $\ell \in R$ off-range)};

  \node[box, below left=0.7cm and 0.05cm of w] (orig)
       {$C_{\mathit{orig}}$\\(contains $\ell$)};
  \node[box, below right=0.7cm and 0.05cm of w] (deb)
       {$C_{\mathit{deb}}$\\(without $\ell$)};

  \draw[->] (w) -- (orig);
  \draw[->] (w) -- (deb);

  \node[below=0.25cm of orig, font=\footnotesize] (orig-r)
       {\textbf{rejects}\,(by constr.)};
  \node[below=0.25cm of deb, font=\footnotesize, align=center] (deb-r)
       {\textbf{rejects}? sound iff yes};

  \draw[->] (orig) -- (orig-r);
  \draw[->] (deb)  -- (deb-r);
\end{tikzpicture}
\caption{One validation iteration.}
\Description{A witness with one perturbed cell is checked against the original
and debloated circuits. The original rejects it by construction; the cross-check
looks for an unsound acceptance by the debloated circuit.}
\label{fig:soundness-fuzz}
\end{figure}
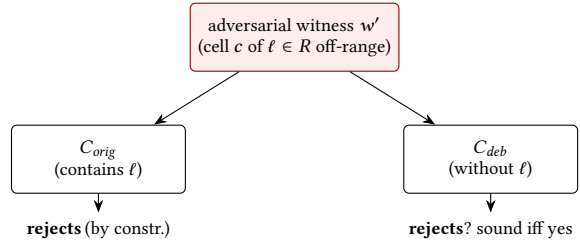

\subsection{Results}
\label{app:soundness-results}

Table~\ref{tab:soundness-fuzz} reports the results. Across all 13 pipeline-model
pairs we executed $3.62 \times 10^{11}$ adversarial iterations and observed zero
divergences between $C_{\mathit{orig}}$ and $C_{\mathit{deb}}$. Coverage of
$R_{\mathit{lookup}}$ was complete on every benchmark: each removed lookup was
tested at least once and, on the smaller zkml side, typically thousands of
times.

\begin{table}[h]
\centering
\caption{Empirical soundness validation on the RQ1 corpus.}
\label{tab:soundness-fuzz}
\small
\setlength{\tabcolsep}{5pt}
\begin{tabular}{l rrr}
\toprule
Model & Removed lookup count & Tests & Unsound \\
\midrule
\multicolumn{4}{l}{\emph{zkml}} \\
\quad MNIST-MLP    &      788   & $1.60{\times}10^{8}$  & 0 \\
\quad LeNet-5      &   9{,}378  & $5.10{\times}10^{9}$  & 0 \\
\quad PTB-RNN      &   3{,}550  & $1.97{\times}10^{8}$  & 0 \\
\quad PTB-GRU      &   9{,}000  & $3.52{\times}10^{10}$ & 0 \\
\quad PTB-LSTM     &   9{,}640  & $2.55{\times}10^{10}$ & 0 \\
\quad BERT-tiny    &   4{,}928  & $5.07{\times}10^{8}$  & 0 \\
\addlinespace[2pt]
\multicolumn{4}{l}{\emph{ezkl}} \\
\quad MNIST-MLP    &   2{,}755  & $4.91{\times}10^{10}$ & 0 \\
\quad LeNet-5      &  38{,}343  & $4.05{\times}10^{10}$ & 0 \\
\quad MobileNet-v2 & 871{,}475  & $5.08{\times}10^{10}$ & 0 \\
\quad PTB-RNN      &  63{,}201  & $3.14{\times}10^{10}$ & 0 \\
\quad PTB-GRU      & 210{,}201  & $2.52{\times}10^{10}$ & 0 \\
\quad PTB-LSTM     & 252{,}001  & $4.92{\times}10^{10}$ & 0 \\
\quad BERT-tiny    & 310{,}593  & $4.91{\times}10^{10}$ & 0 \\
\midrule
\textbf{Total} & \textbf{1{,}785{,}853} & $\mathbf{3.62{\times}10^{11}}$ & \textbf{0} \\
\bottomrule
\end{tabular}
\end{table}

\end{document}